\documentclass[a4paper]{article}

\usepackage[english]{babel}
\usepackage{stix}

\usepackage{array}
\usepackage{algpseudocode}
\usepackage{float}
\usepackage[a4paper,top=3cm,bottom=2cm,left=3cm,right=3cm,marginparwidth=1.75cm]{geometry}

\usepackage{amsmath}
\usepackage{amssymb}
\usepackage{amsthm}
\usepackage{proof}
\usepackage{subcaption}
\usepackage[colorinlistoftodos]{todonotes}
\usepackage[]{hyperref}
\def\rch#1{\rotatebox[origin=c]{180}{#1}}

\DeclareMathOperator*{\E}{\mathbf{E}\,}
\DeclareMathOperator*{\Pb}{\mathbf{P}\,}
\DeclareMathOperator*{\En}{E}

\newcommand{\V}{\mathscr{V}}
\newcommand{\Ed}{\mathscr{E}}

\newcounter{Counter}
\newtheorem{theorem}[Counter]{Theorem}
\newtheorem{remark}[Counter]{Remark}
\newtheorem{definition}[Counter]{Definition}
\newtheorem{lemma}[Counter]{Lemma}

\newtheorem{conjecture}[Counter]{Conjecture}

\newcommand*{\mybox}[2]{\colorbox{#1!30}{\parbox{.93\linewidth}{#2}}}
\definecolor{lightgray}{rgb}{0.7, 0.7, 0.7}

\newcommand*{\myboxx}[2]{\colorbox{#1!30}{\parbox{.63\linewidth}{#2}}}

\newcommand*{\myboxxx}[2]{\colorbox{#1!30}{\parbox{.2\linewidth}{#2}}}

\title{What do QAOA energies reveal about graphs?}
\author{Mario Szegedy}

\date{\today}
\begin{document}
\maketitle

\begin{center}
{\textit{Alibaba Quantum Laboratory, Alibaba Group USA, Bellevue, WA 98004, USA}}
\end{center}

\medskip

\begin{abstract}
Quantum Approximate Optimization Algorithm (QAOA) \cite{farhi2014quantum}
is a hybrid classical-quantum algorithm
to approximately solve NP optimization problems such as MAX-CUT.
We describe a new application area of QAOA circuits: graph structure discovery.
We omit the time-consuming parameter-optimization phase and utilize the dependence of QAOA energy
on the graph structure for randomly or judiciously chosen parameters to learn about graphs. 

In the first part, following up on Wang et al. \cite{wang2018quantum} and Brandao et al. \cite{fern2018fixed}
we give explicit formulas. We show that the layer-one QAOA energy for the MAX-CUT problem for three regular graphs 
carries exactly the information: {\em (\# of vertices, \# of triangles)}. We have calculated our explicit formulas differently from  
\cite{wang2018quantum}, by developing the notion of the $U$-polynomial of a graph $G$.
Many of our discoveries can be interpreted as computing $U(G)$ under various restrictions. 

The most basic question when comparing the structure of two graphs is if they are isomorphic or not. 
We find that the QAOA energies separate all non-isomorphic three-regular graphs up to size 18, all strongly regular graphs up to size 26 and
the Praust and the smallest Miyazaki examples.
We observe that the QAOA energy values can be also used as a proxy to how much graphs differ. Unfortunately, 
we have also found a sequence of non-isomorphic pairs of graphs, for which
the energy gap seems to shrink at an exponential rate as the size grows.
Our negative findings however come with a surprise: 
if the QAOA energies do not measurably separate between two graphs, then 
both of their energy landscapes must be extremely flat 
(indistinguishable from constant), already when the number of QAOA layers is intermediately large.
This holds due to a remarkable uncoupling phenomenon that we have only deduced from computer simulation.

%When running QAOA on families of graphs we
%get "QAOA energy landscapes" where clusters can be observed. These clusters turn out to correspond to interesting parameters, such as the number of triangles.

To compute average values of QAOA energies and their standard deviation for a large number of levels,
we introduce {\em QAOA dynamics}, and discover a tool for studying it: higher order density matrices.
\end{abstract}

\medskip

\begin{center}
{\em Dedicated to the 70th birthday of Laci Babai}
\end{center}

\section{Introduction}

Quantum Approximate Optimization Algorithm (QAOA), introduced by Farhi, Goldstone and Gutmann \cite{farhi2014quantum} in 2014,
is an attempt to gain quantum advantage in solving combinatorial 
optimization problems of the form 
\[
\min_{z\in \{0,1\}^{n}}  C(z) \; = \; \min_{z\in \{0,1\}^{n}} \sum_{\alpha = 1}^{m} C_{\alpha}(z), \;\;\;\;\;\;\; \mbox{$C_{\alpha}$ depends only on $\ell$ coordinates of $z$}
\]
with a quantum machine in the circuit-based (as opposed to annealing-based) computational model.
In this article we shall only be concerned with the special case, where we have an undirected graph $G = (\V(G),\Ed(G))$ and
\begin{equation}\label{ee}
C(z) = \sum_{\langle jk\rangle \in \Ed(G)} C_{\langle jk\rangle}(z), \;\;\;\;\;\;\; C_{\langle jk\rangle}(z) = {1\over 2}(1 + \sigma_{j}^{z}\sigma_{k}^{z} )
\end{equation}
$C = C_{G}$, but we omit the subscript when it does not cause ambiguity. When {\em minimized}, this gives a maximum cut of $G$ via
\[
\mathrm{MAXCUT}(G) = |\Ed(G)| - \min_{z} C(z)
\]
Although the problem is classical, Farhi et al. construct a quantum circuit, called the QAOA circuit,
that spits out good approximate solutions when its parameters are optimized.
The optimization process is known as the QAOA optimization phase. It is a loop where the QAOA circuit is repeatedly run,
its output is evaluated and the circuit parameters are reset.
When the final parameter values are reached, the circuit is run a few more times
to produce a set of candidate assignments to the optimization problem.
Finally, the best of all candidates is chosen.
The algorithm has become the most 
frequently discussed quantum-classical hybrid algorithm \cite{McClean_2016}.
% for a number of reasons:
%
%\begin{enumerate}
%\item It requires relatively few qubits to run.
%\item It adapts annealing algorithms to the circuit-based model.
%\item It is easily implementable on quantum machines and simulators.
%\item There are hints of noise-resilience.
%\end{enumerate}

%Several articles give exposition of the subject, including \cite{}. 
More into the details, the level-$p$ QAOA circuit for a graph $G$ and for parameters $\gamma = (\gamma_{0}, \ldots,\gamma_{p-1})$, $\beta = (\beta_{0}, \ldots,\beta_{p-1})$
computes an $n$-qubit quantum state $|\psi\rangle$, where $n = |\V(G)|$, and
\begin{equation}\label{eq1}
|\psi\rangle  = \prod_{q=0}^{p-1} \left(\prod_{v \in \V(G)} e^{-i\beta_{q} X_{v}} \prod_{\langle jk\rangle \in \Ed(G)} e^{-i\gamma_{q}  
C_{\langle jk\rangle}}  \right) \cdot  {1\over \sqrt{2^{n}}} \sum_{z\in \{0,1\}^{n} } |z\rangle
\end{equation}
The product outside is responsible for making the $p$ levels. Since a single level is composed of two operators (the two products inside),
it is tempting to think that the usual circuit depth of the level-$p$ QAOA circuit is $2p$. Although $\prod_{i \in \V(G)} e^{-i\beta_{j} X_{i}}$ has depth one
when presented as a quantum circuit, the depth of $\prod_{\langle jk\rangle \in \Ed(G)} e^{-i\gamma_{j}  C_{\langle jk\rangle}}$ is typically not one, and
finding the smallest depth implementation
requires graph theory. All 2-qubit gates of the form $e^{-i\gamma_{j}  C_{\langle jk\rangle}}$ commute, but they share qubits. One has to refer to the Vizing theorem
to get a depth $d+1$ rendering of these gates in the worst case (and depth $d$ in the best case), where $d$ is the maximum degree of $G$.
All-in-all we get a depth $(d+2)p$ upper bound for the entire circuit.

Although $C$ is treated as the quantum Hamiltonian $\sum_{z\in \{0,1\}^{n}} C(z) |z\rangle \langle z|$,
the ground energy, $\min \langle \psi | C | \psi \rangle$, is achieved at a classical state because of the diagonal form of the Hamiltonian.
Nevertheless, the QAOA circuit produces a quantum state, and  {\em this is precisely where its strength lies}: it adds a quantum dimension 
to an otherwise classical problem.

The QAOA circuit is also called an ``Ansatz,'' i.e. ``rudiment'' or ``approach,'' meaning that it only gets us started solving the optimization problem
$\max_{z} C(z)$. 
The remaining task is to choose the parameter values $\beta_{j}, \gamma_{j} \;\; (0\le j\le p-1)$, called {\em angle sequences}, in an optimal way.
The optimized state is mathematically guaranteed to converge into the 
subspace spanned by the optimal classical solutions  (\cite{farhi2014quantum}, Section VI), but the rate of convergence may be very slow. For small $p$, even after 
setting the parameters in the most optimal way, we may end up with a low-quality solution.

\medskip

\noindent{\bf Graph structure discovery.} 
We harness the ``quantum dimension'' present in
the QAOA Ansatz in a new way: to obtain information about the structure of the 
input graph, via the quantum object in  (\ref{eq1}), which is supposed to encode main features of the graph. 

\begin{definition}
We denote the state in (\ref{eq1}) with $|\gamma,\beta\rangle$ where $\gamma = (\gamma_{0}, \ldots,\gamma_{p-1})$ and $\beta = (\beta_{0}, \ldots,\beta_{p-1})$.
\end{definition}

Instead of looking to solve the optimization problem in (\ref{ee}), our main goal is to understand
\[
\En(G,\gamma,\beta) =  \langle \gamma,\beta | C | \gamma, \beta \rangle
\]
for {\em random } or {\em arbitrary} $ (\gamma,\beta) \in [0,2\pi]^{2p}$.
We compare energy values for pairs or sets of graphs
for identical degree sequences. Our investigation has started from a conjecture which is still unresolved:

\begin{conjecture}\label{iso-conj}
Let $G_{1}$ and $G_{2}$ be two non-isomorphic graphs. Then for some $p>0$, when 
$(\gamma,\beta)$ is randomly and uniformly chosen from $[0,2\pi]^{2p}$ 
we have
\[
\Pb (E(G_{1},\gamma,\beta) \neq E(G_{2},\gamma,\beta)) = 1 
\]
\end{conjecture}

An analogous statement in the case of boson sampling is proven in \cite{bradler2018graph, schuld2019quantum}.
To turn Conjecture \ref{iso-conj} to an even stronger conjecture of quantum polynomial time graph isomorphism algorithm,
the separation must be at least inverse polynomial for random degree sequences, and $p$ must be polynomial in the size of the graph.

In the first part of the paper we ask: what features of a graph $G$ are encoded in $E(G,\gamma,\beta)$?
We show:

\begin{theorem}\label{triangle}
Let $p=1$ and $(\gamma,\beta) \in [0,2\pi]^{2}$. Then for any cubic graph $G$ with $n$ nodes and $t$ triangles:
\[
E(G,\gamma,\beta) = {3n\over 4} + {3n\over 8} \sin 4\beta\sin 2\gamma \cos\gamma + {3 t \over 8} \sin^{2} 2\beta \sin^{2} 2 \gamma
\]
\end{theorem}

\noindent{\bf $U$-polynomials.} 
We have developed this tool to prove Theorem \ref{triangle}, but it was useful in all our calculations. 
One might view $U$-polynomials of a graphs as 
certain types of tensor-networks or partition functions for certain Ising models. The notion's advantage is a graph theory- friendly language.
We defer all information about $U$-polynomials, including their definition, to the appendix.

\smallskip

\begin{center}
$\ast\;\ast\;\ast$
\end{center}

\smallskip

Unlike in the first part of the paper, where all statements were mathematically verified, in the second part, aside from the last section, we rely on computer experiments.
We show that single random QAOA energies can already distinguish all non-isomorphic members of large classes of graphs such as all 3-regular graphs of size 16. 
Unfortunately, it also seems, that
polynomial time graph isomorphism testing with QAOA is unlikely: 
 the $({\rm Circular \; ladder}(n), \; {\rm Moebius \; ladder}(n))$ family of pairs of graphs
seems to exhibit an exponentially shrinking sequence of average energy gaps.
The average energy gap for two graphs, $G_{1}$ and $G_{2}$ is defined as
\[
\Delta(G_{1},G_{2},p) = \E | E(G_{1},\gamma,\beta) - E(G_{2},\gamma,\beta)| \;\;\;\;\;\;\;\;\;\;\; (\gamma, \beta) \; \mbox{is uniform in}\; [0,2\pi]^{2p}
\]
Our hard graph pairs may also disqualify annealer-based graph isomorphism testers as in  \cite{hen2012solving}, and it could be interesting
to analyze them for boson sampling-based testers as well.

We have found an interesting consequence of small energy gaps: due to a numerical observation
what we call a {\em decoupling phenomenon} (see Section \ref{decoupling}),
small $\Delta(G_{1},G_{2},p)$ at large $p$ implies flat energy landscapes for both graphs.
This in turn has hardness consequences on optimization.

There is a good news too: we have evidence that QAOA energy differences can be useful to detect ``intuitive'' graph similarity. Graph similarity is a measure that
exists between any two graphs. In \cite{shaydulin2019evaluating, shaydulin2019multistart} similarity is rigorously defined as the graph edit distance.
In our article we have avoided the expensive graph edit distance calculations
by replacing it with studying a Markov chain on graphs that makes a small local change at every step. We observe that the average QAOA energy gap between graphs that
are farther in this walk is larger.

\smallskip

\begin{center}
$\ast\;\ast\;\ast$
\end{center}

\smallskip

In the third part we develop methods that mathematically address 
questions raised in the second part. The quantity
\begin{equation}\label{square}
\square(G_{1},G_{2},p) = \E | E(G_{1},\gamma,\beta) - E(G_{2},\gamma,\beta)|^{2} \;\;\;\;\;\;\;\;\;\;\; (\gamma, \beta) \; \mbox{is uniform in}\; [0,2\pi]^{2p}
\end{equation}
is easier to analyze than $\Delta(G_{1},G_{2},p)$, so we will focus on $\square(G_{1},G_{2},p)$. 
We assume that $p$ is large, which lets us focus on:

\smallskip

\noindent{\bf The QAOA dynamics.} Let $G$ be an arbitrary graph.
A QAOA circuit for $G$ with random angles and with increasing depth can be made a Markov chain
on the $2^{|\V(G)|}$ dimensional
complex unit ball, ${\cal S}$, with the transition rule:

\smallskip

\begin{center}
{\em Apply a new level of the QAOA circuit with random angles.}
\end{center}

\smallskip

\noindent After each step
we have updated the statistical ensemble of states on $|\V(G)|$ qubits, which 
seems to weakly converge to a limiting distribution, $\Sigma_{\infty}(G)$, on ${\cal S}$. This gives rise to new graph parameters:

\smallskip

\noindent{\bf The QAOA moments} of $G$ are defined as $\mu_{k}(G) = \int_{\cal S} \;  \langle \psi | C_{G} | \psi \rangle^{k} \; d(\Sigma_{\infty}(G))$.
The decoupling phenomenon, described in Section \ref{decoupling}, gives 
\begin{equation}
\lim_{p\rightarrow\infty}  \; \square(G_{1},G_{2},p) \;\; = \;\; \mu_{2}(G_{1}) + \mu_{2}(G_{2}) - 2\mu_{1}(G_{1}) \mu_{1}(G_{2})
\end{equation}
reducing the gap-question to the calculation of $\mu_{1}$ and $\mu_{2}$ of graphs.
We provide a methodology for such a calculation in the form of {\em higher order density matrices}.
We demonstrate the use of this tool by calculating some first and second QAOA moments
for small graphs.

%\section{Charactertistic Polynomial and ``Sequence Discovery''}
%Before discussing graph structure discovery, we discuss its smaller brother, ``sequence discovery.''

\newpage
\part{Explicit Formulas}
The QAOA research has been taking diverse directions.
A lot of emphasis is put on issues such as 
advantage over classical \cite{farhi2014quantum2, farhi2016quantum, hastings2019classical}, 
noise sensitivity \cite{Alam2019AnalysisOQ},
parameter optimization \cite{zhou2018quantum, shaydulin2019evaluating, shaydulin2019multistart, crooks2018performance}, 
implementation \cite{Peruzzo2014AVE, Guerreschi_2019}, simulation \cite{zhang2019alibaba}.

The QAOA energy values more often than not are calculated by computer simulation. 
There are also exceptions, most notably by Wang et al. \cite{wang2018quantum} and Brandao et al. \cite{fern2018fixed}, who make exact calculations.
We follow their tradition and provide a number of explicit formulas for special cases.

\section{Level-1: All graphs (Wang et al.)}
\label{levelone}

We adopt the following result from Wang et al \cite{wang2018quantum} (modified for our notations and corrected a typo). Let $d_{L}$, $d_{R}$ and $d_{M}$ be the number of nodes that are
connected only to the left, only to the right and to both nodes of an edge $e$ of $G$. Then the level one QAOA energy
$E(e) = \langle \beta,\gamma |C_{e} |\beta,\gamma\rangle$ associated with edge $e$ of $G$ and rotation angles $\beta,\gamma \in [0,2\pi]$ is:

\bigskip

\begin{tabular}{ccc}
\myboxx{lightgray}{
\[
E(e) = E_{d_{L}, d_{R}, d_{M}}  =  {1\over 2} + {1\over 4}\, (X + Y)
\] 
\begin{eqnarray*}
X  & = &   \sin^{2} 2\beta \, \cdot \left( 1 - \cos^{d_{M}} 2\gamma  \right) \cos^{d_{L} + d_{R}}  \gamma  \\
Y  & = &  \sin 4 \beta\, \sin \gamma   \cdot \left(\cos^{d_{L} } \gamma + \cos^{d_{R}} \gamma\right)\,  \cos^{d_{M}} \gamma \\
\end{eqnarray*}
}  & \vspace{0.2in} & \myboxxx{lightgray}{\includegraphics[width=0.2\textwidth]{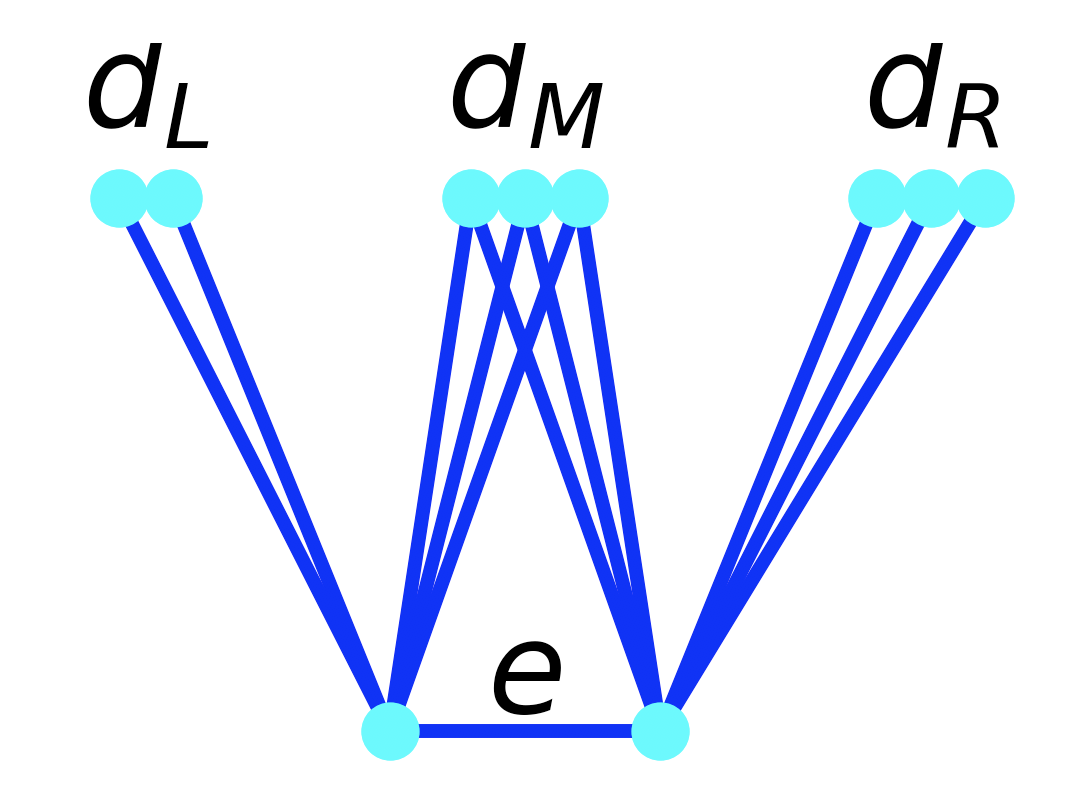} }
\end{tabular}

\medskip

From this we get that the energy of every edge of the cycle $C_{n}$ with $n\le 4$:

\[
E_{cyc} =  E_{1,0,1} = {1\over 2} + {1\over 4} \sin 4 \beta\, \sin 2 \gamma 
\]

More generally, the energy contribution of every edge of a triangle-free $d$-regular graphs is

\[
E_{\Delta{\rm free},d} =  E_{d-1,0,d-1} = {1\over 2} + {1\over 2} \sin 4 \beta\, \sin \gamma \cos^{d-1} \gamma 
\]

yielding ${nd\over 4} + {nd\over 4} \sin 4 \beta\, \sin \gamma \cos^{d-1} \gamma $ when summing it up for all edges.

\section{The Triangle Theorem}

\begin{figure}[H]
\centering
\begin{tabular}{ccccc}
\includegraphics[width=0.1\textwidth]{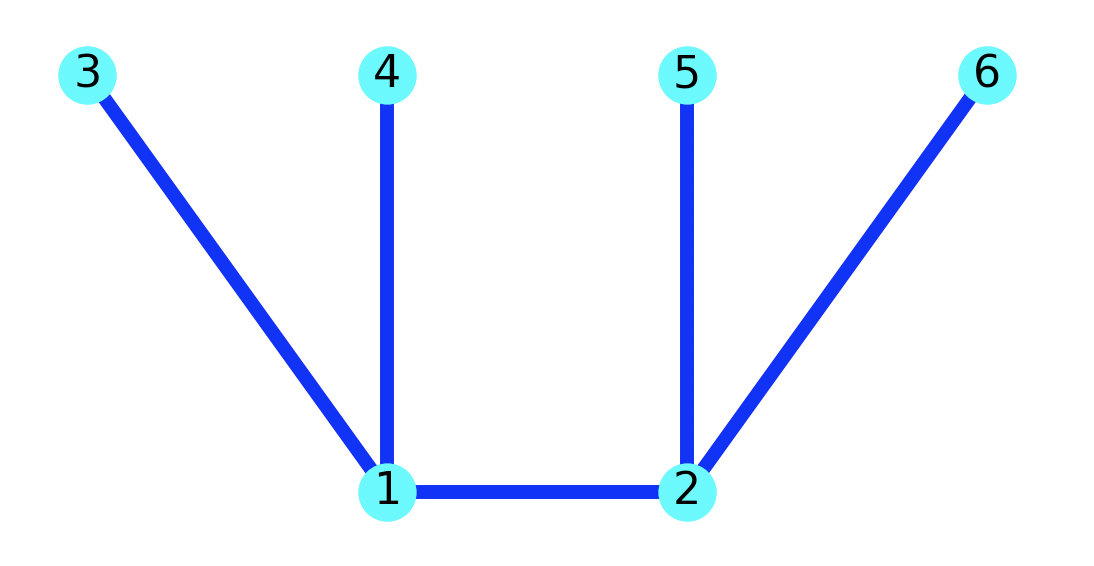} & \hspace{0.05in}  & \includegraphics[width=0.1\textwidth]{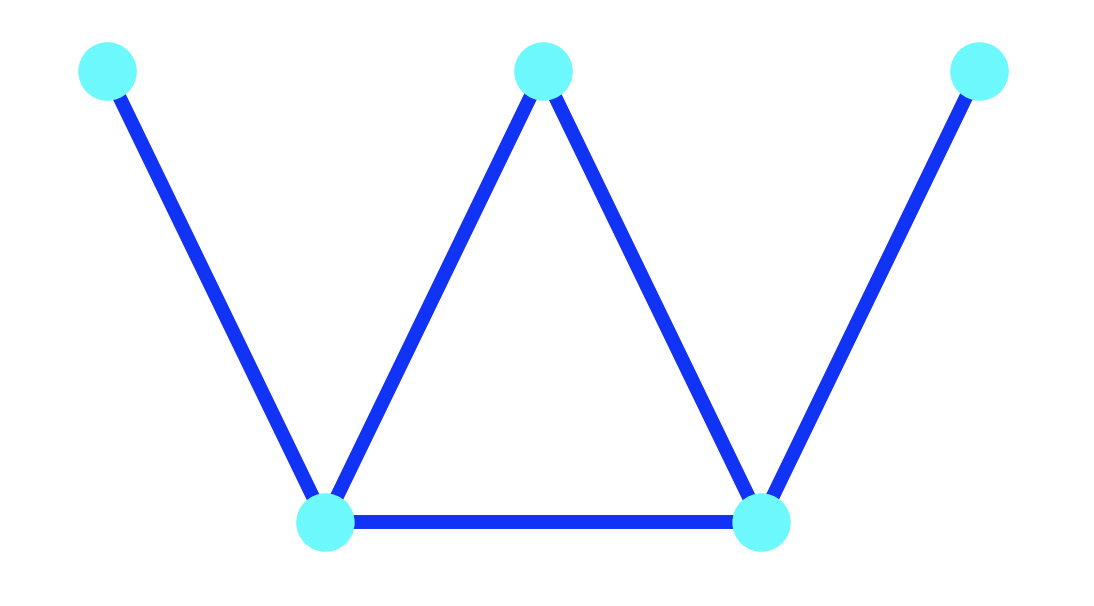}  & \hspace{0.05in}  & \includegraphics[width=0.1\textwidth]{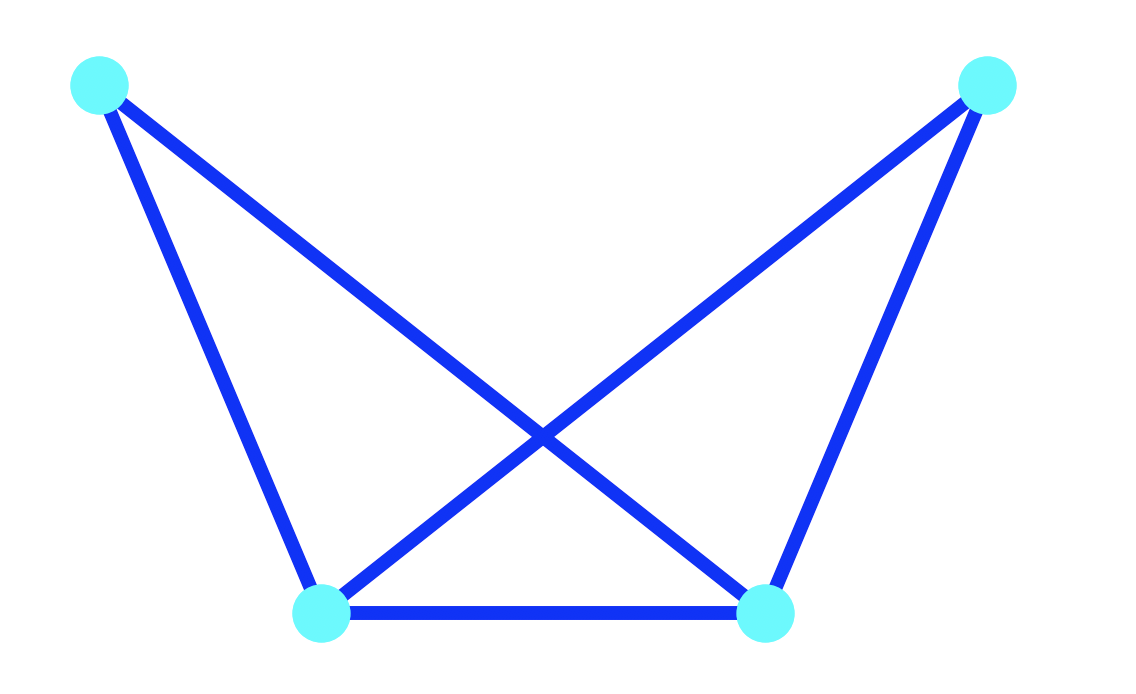}  \\
type 0 & & type 1 & & type 2
\end{tabular}
\caption{The neighborhood types of the edge in the middle for three-regular graphs, level 1 QAOA .\label{edge_neighbor}}
\end{figure}

In this section we prove Theorem \ref{triangle}, which we call the {\em Triangle Theorem}, which states that the level one QAOA
energy of a three regular graph $G$ is:

\begin{equation}\label{tri}
E(G,\gamma,\beta) = {3n\over 4} + {3n\over 8} \sin 4\beta\sin 2\gamma \cos\gamma + {3 t \over 8} \sin^{2} 2\beta \sin^{2} 2 \gamma
\end{equation}

Since $G$ is three regular,
we may encounter three different types of neighborhoods at {\em edge-distance} one as in Figure \ref{edge_neighbor}.
%, and correspondingly three different kinds of energies when calculating the one-level QAOA.
Let $m_{0}$, $m_{1}$ and $m_{2}$ count the number of edges of $G$ with these neighborhood types.
These three parameters already determine the level one QAOA energy of cubic graphs as noted in \cite{fern2018fixed}.
Let the single edge energy associated with type $i$ be $E_{i}$. Then the total energy is
\begin{equation}\label{EQB}
E = m_{0}E_{0} + m_{1}E_{1} + m_{2}E_{2}
\end{equation}

We replace $m_{0}$, $m_{1}$ and $m_{2}$ with just $n$ and $t$. That we can do this is due to a co-incidence:
\begin{lemma}
$E_{0} + E_{2} = 2 E_{1}$ 
\end{lemma}
\begin{proof}
From the displayed expression of the previous section we can express $E_{0}$, $E_{1}$ and $E_{2}$:

\begin{eqnarray}\label{eqn4}
E_{0} = E_{2,0,2} & = & {1\over 2} + {1\over 2} \sin 4 \beta\, \sin \gamma \cos^{2} \gamma \\\label{eqn5}
E_{1} = E_{1,1,1} & = & {1\over 2} +  {1\over 4} \sin^{2} 2\beta  ( 1 - \cos 2\gamma ) \cos^{2}  \gamma + {1\over 2}  \sin 4 \beta\, \sin \gamma \cos^{2} \gamma \\
E_{2} = E_{0,2,0} & = & {1\over 2} +   {1\over 4} \sin^{2} 2\beta  ( 1 - \cos^{2} 2\gamma ) + {1\over 2}  \sin 4 \beta\, \sin \gamma \cos^{2} \gamma
\end{eqnarray}

To verify $E_{0} + E_{2} = 2 E_{1}$ all we have to show is that 
\[
2  ( 1 - \cos 2\gamma ) \cos^{2}  \gamma = 1  -   \cos^{2} 2\gamma 
\]
which can be easily seen from that both sides are $4 \sin^{2}\gamma \cos^{2}\gamma$.\end{proof}

Now we can eliminate $E_{2}$ from Equation (\ref{EQB}):
\begin{equation}\label{EQB3}
E = m_{0}E_{0} + m_{1}E_{1} + m_{2} (2E_{1} - E_{0}) 
\end{equation}
By counting the edges of all triangles in two different ways we get:
\begin{equation}\label{EQA}
3t = m_{1} + 2m_{2}
\end{equation}
which together with Equation (\ref{EQB3}) gives us
\[
E = (m_{0} - m_{2}) E_{0} + 3t E_{1} = (|G| - m_{1} - 2 m_{2}) E_{0} + 3t E_{1}  = \left({3n\over 2} - 3t\right) E_{0} + 3t E_{1} 
\]
since the number of edges of $G$ is ${3n / 2}$, as $G$ is 3-regular. Now Equations (\ref{eqn4}) and (\ref{eqn5}) 
and basic trigonometric identities immediately give Theorem \ref{triangle}.

\section{Level Two: Cycles}\label{leveltwo}

The MAX-CUT QAOA expressions for level-2 have significant complexity.
We have calculated the expression of the 
single edge-energy for cycles of length at least 6 to get an idea about
its form and complexity. We could put the expression into other equivalent forms 
but they were not simpler. In \cite{wang2018quantum} a significantly more complicated formula is given.

\medskip

\begin{center}
{The 2-level QAOA energy for $e$ in $\bullet-\bullet-\bullet\stackrel{e}{-}\bullet-\bullet-\bullet$ with angle sequence $\beta_{0},\beta_{1},\;\gamma_{0},\gamma_{1}$}:

\smallskip
\mybox{lightgray}{
\[
{1\over 2} + {1\over 4} (X + Y + Z + W)
\]
\begin{eqnarray*}
X & = & - {1\over 2} \sin 2 \gamma_{1}  \cos 2 \gamma_{0}\cdot\, \left( \sin^{2} 2\beta_{1} \sin 4\beta_{0} - 2 \sin 4\beta_{1} + \sin 4\beta_{1} \sin^{2} 2\beta_{0} \right) \\[4pt]
Y & = &   - {1\over 2} \cos 2 \gamma_{1} \sin 2 \gamma_{0}\cdot\, \left( \sin^{2} 2\beta_{1} \sin 4\beta_{0} - 2 \sin 4\beta_{1} + \sin 4\beta_{1} \sin^{2} 2\beta_{0} \right) \\[4pt]
Z & = &   \sin 2 \gamma_{0} \sin 4\beta_{0} \left({1\over 4} + {3 \over 4} \cos 4 \beta_{1} \right) \\[4pt]
W & = &  \sin 2 \gamma_{1} \sin 2\beta_{1} \sin 2\beta_{0} \sin 2 (\beta_{1} + \beta_{0} )
\end{eqnarray*}}
\end{center}

\smallskip

The QAOA energy of a $n$-cycle $n\ge 6$ is simply $n$ times the above amount.

\section{Analysis of the single-edge graph}

Formulas for level two QAOA are complex,
but level $p$ seems nearly intractable. Here even the graph containing a single edge is a challenge. 
The straightforward state evolution for a $\gamma,\beta$ sequence gives the final state
\[
|\phi\rangle =  N(\beta_{p-1}) M(\gamma_{p-1}) \cdots  N(\beta_{0}) M(\gamma_{0})\, |+\rangle^{\otimes 2}
\]
where
\begin{eqnarray*}
N(\beta)  &= &
\left(
\begin{array}{cc}
\cos \beta & - i \sin \beta \\
 - i \sin \beta & \cos \beta \\
\end{array}
\right) \otimes
\left(
\begin{array}{cc}
\cos \beta & - i \sin \beta \\
 - i \sin \beta & \cos \beta \\
\end{array}
\right)
 \\[8pt]
  &= & \left(
\begin{array}{cccc}
\cos^{2} \beta & -{ i\over 2}  \sin 2 \beta & -{ i\over 2}  \sin 2 \beta & - \sin^{2} \beta  \\[2pt]
 -{ i\over 2}  \sin 2 \beta & \cos^{2} \beta & - \sin^{2} \beta & -{ i\over 2}  \sin 2 \beta \\[2pt]
  -{ i\over 2}  \sin 2 \beta & - \sin^{2} \beta & \cos^{2} \beta & -{ i\over 2}  \sin 2 \beta \\[2pt]
  - \sin^{2} \beta & -{ i\over 2}  \sin 2 \beta & -{ i\over 2}  \sin 2 \beta & \cos^{2} \beta
\end{array}
\right)
\end{eqnarray*}

and $M(\gamma) = {\rm Diag}(e^{-i\gamma}, 1, 1, e^{-i\gamma})$.
Once the state $|\psi\rangle$ is iteratively computed for level $p$, 
the energy is $||\psi\rangle_{00}|^{2} + ||\psi\rangle_{11}|^{2}$.
Interestingly, we can come up with a different (although similar) formula that avoids taking the squares
when computing the energy.
Let 
\[
v_{0} = (0.5, 0, 0, 0)
\]
be a starting vector. Define Matrices
\begin{eqnarray*}
M_{1} = 
\left(
\begin{array}{cc}
1 & 0 \\
 0 & e^{-i  \gamma } \\
\end{array}
\right)
& \;\;\;\;\;\; &
N_{1} = 
\left(
\begin{array}{cc}
\cos 2 \beta & i \sin 2 \beta \\
 i \sin 2 \beta & \cos 2 \beta \\
\end{array}
\right)
\\[5pt]
M_{2} = 
\left(
\begin{array}{cc}
1 & 0 \\
 0 &  e^{i  \gamma } \\
\end{array}
\right)
& \;\;\;\;\;\; &
N_{2} = 
\left(
\begin{array}{cc}
\cos 2 \beta & - i \sin 2 \beta \\
 -  i \sin 2 \beta &  \cos 2 \beta \\
\end{array}
\right)
\end{eqnarray*}

Then the energy value can be expressed as

\[
E(\gamma,\beta)  = (1,1,1,1)^{T} \left(\prod_{i=0}^{p-1} M_{1} (\gamma_{p-i})N_{1}(\beta_{p-i}) \otimes M_{2}(\gamma_{p-i}) N_{2}(\beta_{p-i}) \right) v_{0}
\]

It is worthwhile to write out $M(\gamma,\beta) = M_{1} (\gamma)N_{1}(\beta) \otimes M_{2} (\gamma)N_{2}(\beta)$ (for $\beta,\gamma \in [0,2\pi]$).

\smallskip

\[
M(\gamma,\beta)  \; = \;
 \left(
\begin{array}{cccc}
\cos^{2} 2 \beta & -{ i\over 2}  \sin 4 \beta & { i\over 2}  \sin 4 \beta &  \sin^{2} 2 \beta  \\[2pt]
 -{ i\over 2}  e^{i\gamma} \sin 4 \beta &  e^{i\gamma} \cos^{2} 2 \beta & e^{i\gamma}  \sin^{2} 2 \beta & { i\over 2}   e^{i\gamma}  \sin 4 \beta \\[2pt]
  { i\over 2}  e^{-i\gamma}  \sin 4 \beta & e^{-i\gamma}   \sin^{2} 2 \beta & e^{-i\gamma}   \cos^{2} 2 \beta & -{ i\over 2}  e^{-i\gamma}  \sin 4 \beta \\[2pt]
  \sin^{2} 2 \beta & { i\over 2}  \sin 4 \beta & -{ i\over 2}  \sin 4 \beta & \cos^{2} 2 \beta
\end{array}
\right)
\]

\smallskip

It is easy to see that 

\[
\E \left(
M(\gamma,\beta) 
\left(
\begin{array}{c}
a \\
b \\  
c \\
d
\end{array}
\right)\right) \; = \; 
\left(
\begin{array}{c}
{a + d\over 2} \\
0 \\  
0 \\
{a + d \over 2}
\end{array}
\right)
\;\;\;\;\;\;\;\;\; {\rm over}\; (\gamma,\beta) \in U([0,2\pi]^{2}) 
\]
We can repeatedly apply the above, starting from $v_{0}$, to get:
\begin{lemma}\label{edgesigma}
The expected QAOA energy of a graph containing a single edge is ${1\over 2}$ for every $p$.
\end{lemma}

The lemma does not imply, that the average QAOA energy of every graph $G$ is $|\Ed(G)|/2$.
Experiments show that $|\Ed(G)|/2$ is not always the average energy, but it is a good approximation.

\section{Examples}

\begin{figure}[H]
\centering
\begin{tabular}{ccc}
\includegraphics[width=0.3\textwidth]{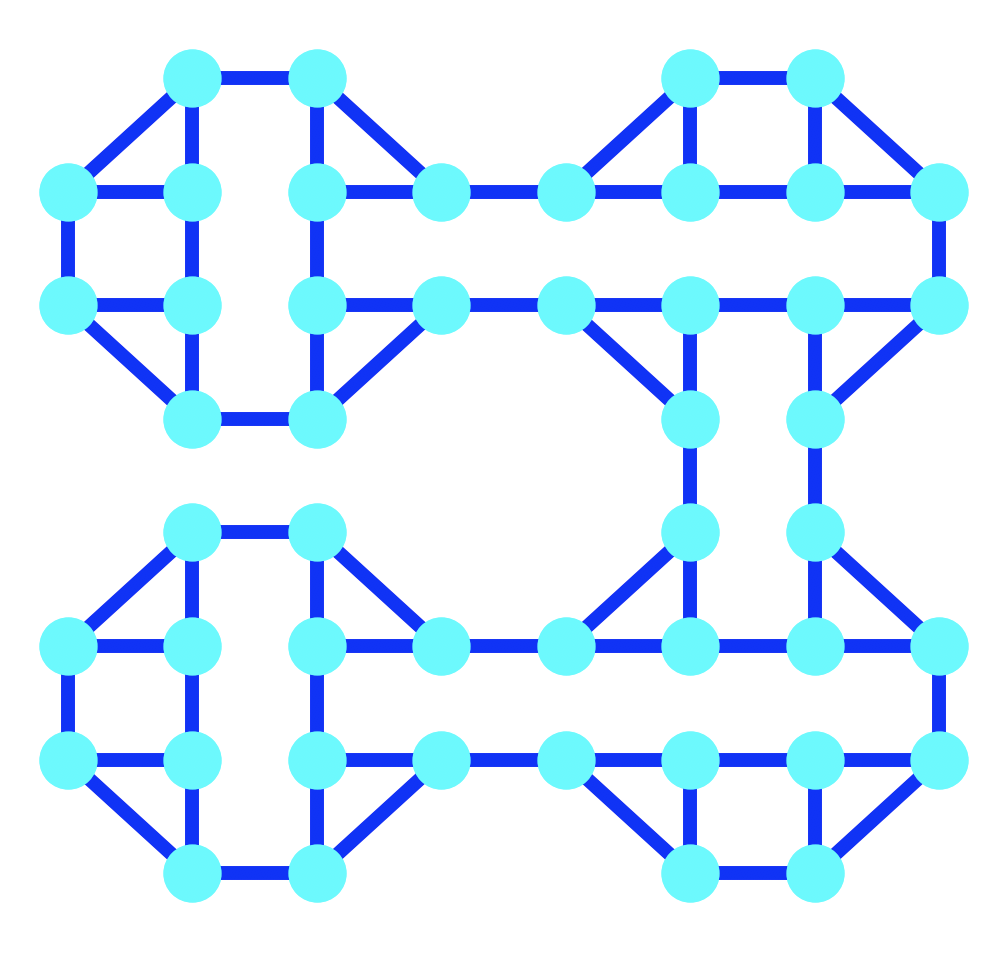} & \hspace{0.3in} & \includegraphics[width=0.3\textwidth]{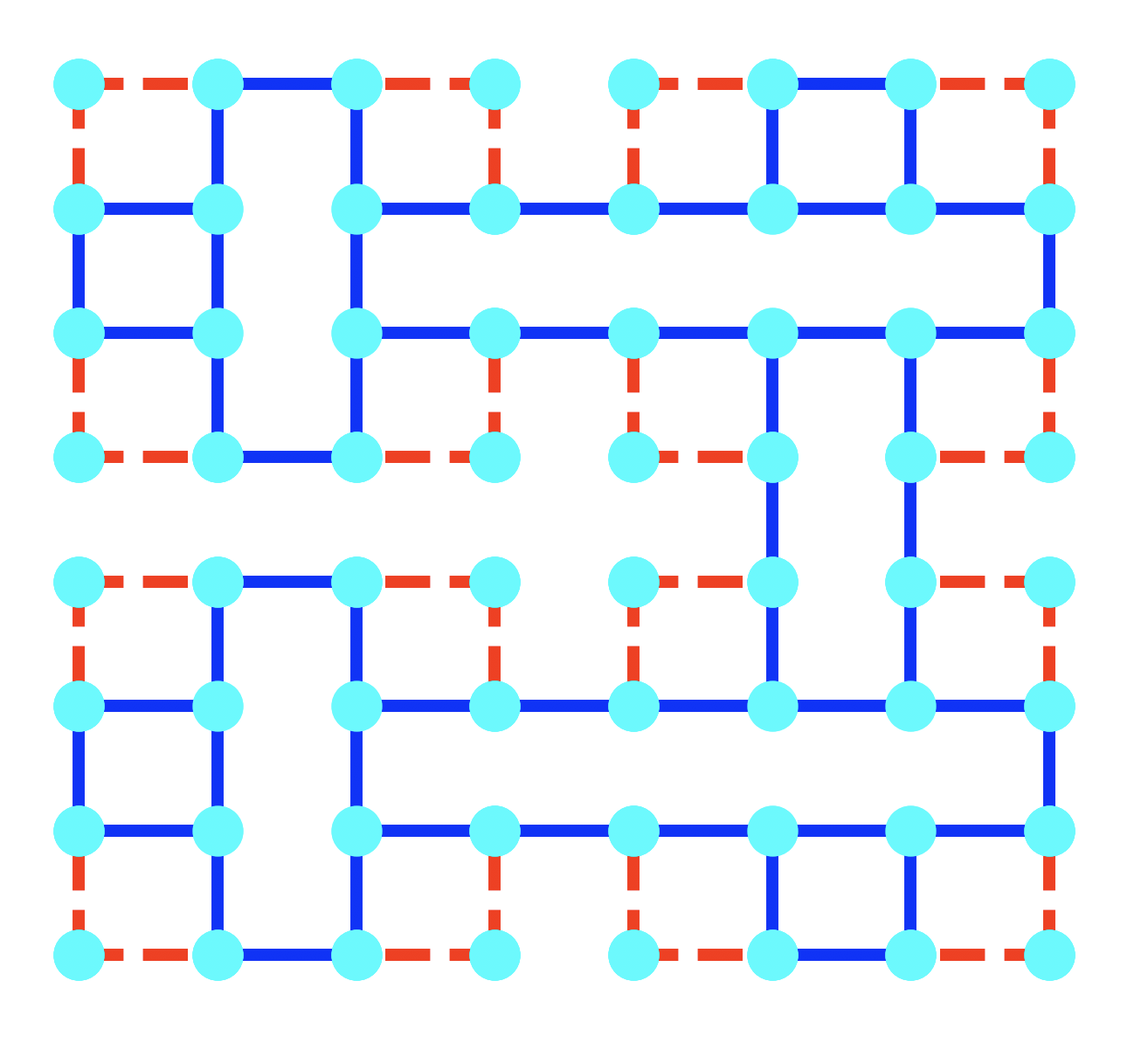} \\
  a.) Test graph $\rch{C}_{48}$ & & b.) Grid implementation of $\rch{C}_{48}$
\end{tabular}
\caption{Testing level one QAOA  on an 8 by 8 grid\label{grid}} 
\end{figure}

Assume we want to validate the QAOA algorithm on a quantum chip with an 8 by 8 grid architecture.
Since this is our first test of the device, we want to choose a graph $G$ with low degree to reduce the circuit depth. The architecture for instance allows to embed
a cycle of length 64 without intersection. We use the formula in Section \ref{levelone} to get the level one QAOA energy
\[
32 + 16 \sin 4\beta \sin 2\gamma
\]
Another example is the graph $\rch{C}_{48}$ shown in Figure \ref{grid}/a, which is 3-regular.
 The nice thing about $\rch{C}_{48}$ is that it can be implemented 
 on the 8 by 8 grid as shown in Figure \ref{grid}/b. For a diagonal edge $\langle vw\rangle$ the operation 
 $e^{-i\gamma  C_{\langle vw\rangle}}$ 
can be performed as
 \[
 e^{-i\gamma  C_{\langle vw\rangle}} = {\rm SWAP}(xv) \cdot e^{-i\gamma  C_{\langle xw\rangle}}  \cdot {\rm SWAP}(xv) 
 \]
 where $v,x,w$ form a right triangle (the grid edges replacing the diagonal edges are dashed red in Figure \ref{grid}/b).
 The graph $\rch{C}_{48}$ has 48 nodes and 16 triangles, therefore its level one QAOA energy is 
 \[
 36 + 18 \sin 4\beta\sin 2\gamma \cos\gamma + 6 \sin^{2} 2\beta \sin^{2} 2 \gamma
 \]
 Further, if we set $\gamma = \beta = \pi/4$ the formula gives 42. In contrast, a random assignment would give energy 72/2 = 36, on expectation.

\newpage
\part{Graph Similarity and Graph Isomorphism}

\section{Algorithm}\label{algsection}

Our fundamental algorithm computes the QAOA energies of a set $\{G_{1},\ldots, G_{k}\}$ 
of graphs (often just $G_{1}$ and $G_{2}$) with respect to the {\em same} random degree sequences
\begin{eqnarray*}
\gamma= (\gamma_{0},\ldots,\gamma_{p-1}) \in \; U[0,2\pi]^{p} \\
\beta= (\beta_{0},\ldots,\beta_{p-1}) \in \; U[0,2\pi]^{p} 
\end{eqnarray*}
If two graphs are isomorphic, their energies are the same for the same $(\gamma, \beta)$.
Our experiments indicate that if two graphs are not isomorphic, their energies differ for some large enough $p$.
The indication is admittedly weak:  we have not found any counter-example in spite of probing large families of graphs
as well as some specific hard pairs.
The smallest $p$, which separates all graphs on $n$ nodes,
which is not ruled out by our trials is $n/4$. It would be surprising 
if more levels than a small constant times the number of edges (or even nodes) was necessary 
to separate any pair of connected graphs, if they are separable at all.
Conjecture \ref{iso-conj} expresses our belief that QAOA energies distinguish non-isomorphic graphs. This 
is one of our motivating questions. An even bigger question is if the energy gaps are large enough to be detectable
with a quantum computer.
As we shall see, we have counter-indications for that. 

For our experiments we have computed the QAOA energies of all graphs with a classical simulator.
Double precision was sufficient for us, although theoretically, with a polynomial factor overhead, we could have afforded 
computing all values with polynomially many digits of precision.
In classical simulations the bottleneck is not the precision, but that the number of arithmetic operations grows
exponentially with $|\V(G)|$ even for polynomial $p$. As long as we are
satisfied with logarithmic digits of precision, i.e. with an additive $\epsilon =1/{\rm poly}(|\V(G)|)$ output error,
an estimator $\tilde{E}(G,\gamma, \beta)$ of ${E}(G,\gamma, \beta)$
can be computed with a quantum computer for graphs with $n$ nodes and with $p = O({\rm poly}(n))$ levels.
This is stated, among others, in \cite{farhi2016quantum}, and the underlying algorithm is really simple:

%\begin{algorithmic}
 %\If {Number of edges of $G_{1}$ $\neq$ number of edges of $G_{2}$}
 %\State Output: ``Not isomorphic''
 %\Else
% {\State $p \gets$  Number of edges of $G_{1}$
 %\vspace{2pt}\State  $(\beta,\gamma) {\gets} \; {\rm uniform\; in} \; [0,2\pi]^{2p}$
%\vspace{2pt} \hspace{0.3in} \State $E_{1}\gets \tilde{E}(G_{1},\beta,\gamma)$
%\vspace{2pt}\hspace{0.3in} \State $E_{2}\gets \tilde{E}(G_{2},\beta,\gamma)$
%\vspace{2pt} \If {$|E_{1} - E_{2}| < 1/p$}
%\vspace{2pt}   \State Output: ``Isomorphic''
%\Else
%        \State Output: ``Not isomorphic''
%\vspace{2pt}  \EndIf}
% \EndIf
%\end{algorithmic}
%
%\medskip

\medskip

{\tt
\noindent...............................................................
\begin{algorithmic}
\State {\bf Input} $G,\gamma,\beta$ 
\For {$1\le j \le  N$}
\State Build a QAOA circuit for $G,\gamma,\beta$ 
\State $z \gets$ Measure the output state in the computational basis
\State $E_{j} \gets$ Classically compute $C(z)$
\EndFor
\State {\bf Return} $\tilde{E}(G,\gamma,\beta) \; \gets $ average of  $E_{j}$s
\end{algorithmic}
...............................................................}

\medskip

Here $N$, the number of repetitions, is a parameter of the algorithm, which must be sufficiently large.
If the fidelity of the circuit-output $|\psi\rangle$ is $1-\epsilon$, then
for an actual output $|\psi'\rangle$ we have $|\langle\psi'|\psi\rangle|^{2} = 1- \epsilon$. Then
$|\langle\psi'|C|\psi'\rangle -\langle\psi|C|\psi\rangle| =  |\langle\psi' -\psi |C|\psi'\rangle +    \langle\psi|C|\psi' -\psi\rangle| \le 2 |\psi' -\psi| |C| \le 2 n^{2} \epsilon$, when $|\V(G)| = n$,
showing that a circuit with fidelity inverse polynomially close to 1 still works sufficiently well.

\section{The Cost of Communicating Isomorphism of Graphs}

Consider the problem where Alice gets a graph $G_{1}$, Bob gets $G_{2}$, both on $n$ nodes, 
and they want to find out if $G_{1}$ is isomorphic to $G_{2}$.
If the only resource we care about is the communication cost between Alice and Bob, there is a constant bit 1\% error 
protocol in the {\em public coin} setting, where a random string is given to both Alice and Bob (at no cost) before their exchange begins.

We will show how to reduce the problem to the following famous public coin communication protocol for the Equality function:
\[
{\rm EQ}(X,Y) =
\left\{\begin{array}{lll}
0 & {\rm if} & X=Y \\
1 & {\rm if} & X\neq Y
\end{array}\right. \;\;\;\;\;\;\;\;\; X, Y \in \{0,1\}^{n}
\]
In this protocol both Alice and Bob get a random string $Z\in \{0,1\}^{n}$. Then Alice sends over
the modulo two inner product of $X$ and $Z$ to Bob, who in turn outputs 1 if $(X,Z)=(Y,Z)$ and 0 otherwise. 
The protocol always succeeds when the two strings are equal, and if not, it will be revealed with probability 0.5.
With a constant number of repetitions the probability of failure can be reduced to below 1\%.

At first the graph isomorphism problem seems much harder, since Alice and Bob have to deal with an unknown isomorphism 
between their respective input graphs. 
The predicate  ${\rm EQ}(G_{1}, G_{2})$ can only reveal if $G_{1}$ and $G_{2}$ are {\em written down in the same exact way}.
There is a way however to get around this problem.
The idea is that Alice and Bob first independently bring their graphs
into their respective {\em canonical forms}. A canonical form is a map ${\cal C}$ from graphs to strings with the property that for two graphs, $G_{1}$ and $G_{2}$
we have ${\cal C}(G_{1}) = {\cal C}(G_{2})$ if and only if the two graphs are isomorphic.
Since graph isomorphism is an equivalence relation, such a map exists.
The protocol fixes this map, and Alice and Bob are left to solve ${\rm EQ}({\cal C}(G_{1}), {\cal C}(G_{2}))$ 
with constant bits of communication (the communication is constant even if ${|\cal C}(G_{i})|$ is exponential).

The problem becomes much harder if we are also concerned with the cost of computing ${\cal C}$. Laci Babai 
in a recent work has given a function ${\cal C}$ computable in time $2^{{\rm polylog} n}$ \cite{DBLP:conf/stoc/Babai19}.
Such a function also solves the graph isomorphism problem in time $2^{{\rm polylog} n}$,
but the converse is not straightforward, and in fact three years have elapsed between results \cite{DBLP:conf/stoc/Babai19}
and \cite{DBLP:conf/stoc/Babai16},
where the first quasi-polynomial graph isomorphism algorithm was introduced.

The QAOA Ansatz for the MAXCUT problem of $G$ offers a way 
to construct a {\em randomized} map ${\cal Q}: \; G\rightarrow E(G,\gamma, \beta)$ 
that canonically encodes graphs, conditional to Conjecture \ref{iso-conj}, when $p$ is sufficiently large.
The associated protocol is:

{\tt
\noindent...............................................................
\begin{algorithmic}
\State Alice and Bob get the same random $\beta,\gamma$ (public randomness)
\State Alice $\leftarrow G_{1}$, Bob $\leftarrow G_{2}$
\State Alice computes ${E}(G_{1},\gamma,\beta)$ 
\State Bob computes ${E}(G_{2},\gamma,\beta)$ 
\State {\bf Return} ``isomorphic'' if ${E}(G_{1},\gamma,\beta) = {E}(G_{2},\gamma,\beta)$, else ``not isomorphic''
\end{algorithmic}
...............................................................}

\medskip

In most {\em (but unlikely in all)} cases, ${\cal Q}$ separates two non-isomorphic graphs 
with high probability {\em even} if we do the computation only with {\em logarithmically many} bits of precision.
In such cases we can replace ${E}(G,\gamma,\beta)$ with estimator $\tilde{E}(G,\gamma,\beta)$, computed by 
the quantum algorithm of the previous section. Further, the energy values $\tilde{E}(G_{1},\gamma,\beta)$ and $\tilde{E}(G_{2},\gamma,\beta)$  need not be 
further composed with ${\rm EQ}$ if we are satisfied with logarithmic bits of communication between Alice and Bob. 
The QAOA map ${\cal Q}$ {\em itself behaves as a kernel function},
something like the inner product in the EQ protocol.
The following analogy might be enlightening:

\medskip

\noindent{\em Alice is given $X\in [m]^{n}$ and Bob is given $Y\in [m]^{n}$. Construct a time- and communication- efficient 
randomized communication protocol that finds out if
$X$ and $Y$ contain the same elements of $[m]$, each the same number of times}

\medskip

In the trivial solution Alice and Bob privately sort their input with multiple occurrences kept (hence both computing ``the'' canonical form
of their respective sequences) and then apply the EQ protocol for the sorted sequences.
There is a different protocol however,
where a random $\gamma\in [0,2\pi]$ is used. Alice sends over $\sum_{i=1}^{n} e^{-i \gamma X_{i}}$ to Bob, who then 
compares this with $\sum_{i=1}^{n} e^{-i \gamma Y_{i}}$, and outputs 1 if the two numbers are equal. 
We consider this protocol as the smaller brother of our QAOA based algorithm for graph isomorphism.
It turns out that the above protocol for identifying multi-sets already performs well, when the numbers are calculated with $O(\log nm)$ 
digits of precision. 

Further, it is not hard to see a connection between the above formulas and characteristic functions of probability theory. It would be worthwhile 
to explore if probability theory could take any use of QAOA energies or related formulas.

\section{Graph Isomorphism and Quantum}

Utilizing the power of quantum 
in the context of the graph isomorphism (GI) problem
has been put forward in many works.
Since the GI problem shares some common traits 
with the integer factoring problem, researchers sense here yet another spectacular
demonstration of quantum advantage.
Nevertheless, there are reasons to be cautious. Although both GI and 
factoring can be viewed as special cases of the hidden subgroup problem,
the two problems behave differently. The GI problem (classically) is very easy on average:
An early result of L. Babai, 
P. Erd\H{o}s and S. M. Selkow \cite{DBLP:journals/siamcomp/BabaiES80}
shows that a straightforward linear time canonical labeling algorithm applies to almost all graphs.
The worst case classical complexity of GI is ${\rm GI\in DTIME}(2^{\log^{O(1)} n})$ \cite{DBLP:conf/stoc/Babai16}, while
factoring seems to require time $2^{n^{\alpha}}$ for some $\alpha > 0$, possibly $\alpha = 1/3$. 
Interestingly, while GI is relatively easy classically, it seems to resist the hidden subgroup problem approach, while the harder factoring yields to it.

A different quantum approach to GI is to look for quantum-computable graph invariants \cite{PhysRevA.100.052317, Zhao_2016}.
A prospective way of obtaining a full set of graph-invariants, i.e. an array of graph parameters that separate between any two non-isomorphic graphs, is via quantum walks.
In a broader sense both \cite{PhysRevA.100.052317} and QAOA qualify as $n$-particle quantum walks.
In \cite{Gamble_2010} 
it is experimentally shown that quantum walks of two interacting particles can successfully distinguish between
some strongly regular graph pairs 
that single particle walks or non-interacting particle walks provably cannot.
Similar results were obtained by S. D. Berry and J. B. Wang
\cite{berry2011two}. Godzil and Guo \cite{godsil2011quantum} calculate spectra of quantum walk operators,
but they do not conclude that two interactive particle walks always distinguish non-isomorphic strongly regular graphs. J. Smith has a 
publication in the arXiv entitled ``k-Boson quantum walks do not distinguish arbitrary graphs'' \cite{smith2010k}.
I. Hen and A.P. Young have experimentally tested a quantum annealing based graph isomorphism tester \cite{hen2012solving}
with some satisfying outcomes, but the authors also express:
``The results we presented here support a conjecture that the Quantum Adiabatic prescription can differentiate between 
all non-isomorphic graphs, given an appropriate choice of problem and driver Hamiltonians. 
This conjecture needs to be tested more thoroughly, both theoretically and also by experiments on real quantum annealers.''
We think, that if our algorithm does not work, there is little chance the annealing based algorithm will,
since QAOA was distilled from the former. The same pairs may fool both.
The work of D. Tamascelli and L. Zanetti \cite{Tamascelli_2014}  is different from the previous ones, in that 
it starts with an equivalent rewriting of the graph isomorphism problem as an optimization problem,
so there is a guarantee that their algorithm succeeds. The question is however the running time,
that can easily be exponential.

In a sequence of innovative works by Xanudu reserchers K. Bradler, S. Friedland, J. Izaac, N. Killoran, and D. Su
\cite{bradler2018graph} and later M. Schuld, K. Bradler, R. Israel, D. Su, and B. Gupt
\cite{schuld2019quantum} the authors identify graphs encoded in quantum state of light.
They show that photon states encoding non-isomorphic graphs give different detection probabilities.
What remains is to upper bound the size of their sampler and to lower bound the 
gap in the statistics.

\section{Graph Similarity and QAOA energies}

\begin{figure}[H]
\centering
\includegraphics[width=0.15\textwidth]{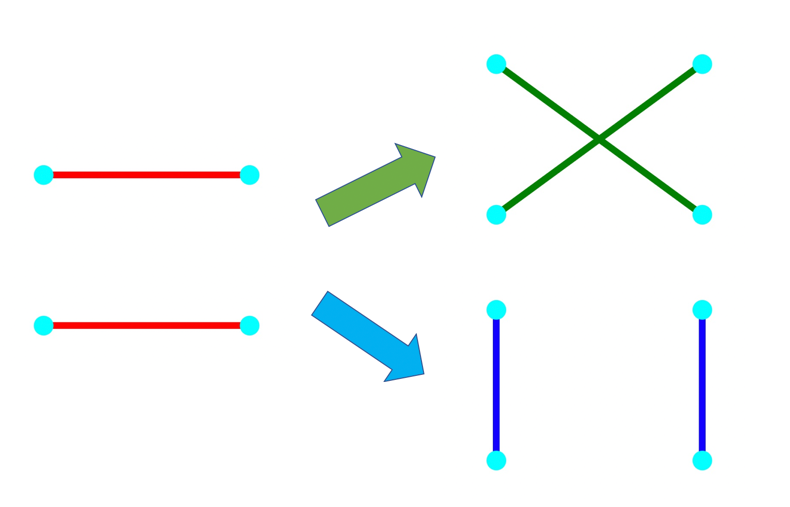} 
\caption{A walk step  \label{walk_step}}
\end{figure}

Even if our QAOA-based algorithms fail to detect graph isomorphism, they might still recognize if two graphs are similar.
What is graph similarity?

The QAOA energy-gap dependence from {\em graph edit distance} between pairs of graphs was studied 
in \cite{shaydulin2019evaluating, shaydulin2019multistart}. Isomorphic graphs have graph edit distance 
zero. In general, graph edit distance 
between graphs $G_{1}$ and $G_{2}$ is the length of the shortest add-delete sequence of edges
that takes $G_{1}$ to some isomorphic copy of $G_{2}$. This distance is a metric.
Shaydulin et al. have looked at how the maximal QAOA energy 
and optimal angle sequences are different for graphs that are close in graph edit distance. 
They have found that graph edit distance is a good predictor
whether these differences are small or large. Our pursuits differ from 
the above research in three ways:

\begin{enumerate}
\item Our goal is not to find or estimate optimal QAOA angles, but rather
to use QAOA to detect graph similarity.
\item We do not compare energies at optimal angles, but rather at random angles.
\item For our experiments we do not go through the hard task of computing graph edit distances. Rather, 
we are satisfied with a more heuristic method.  We have designed a random walk on $d$-regular graphs 
such that each step changes 
the graph edit distance by at most four. We expect that starting from any graph, as
we walk away from it and plot the QAOA energy differences, they grow.
\end{enumerate}

We have found energy difference that increases with the walk-length and eventually reaches a plateau,
which is probably due to the walk mixing into random graphs.

\begin{figure}[H]
\centering
\includegraphics[width=0.25\textwidth]{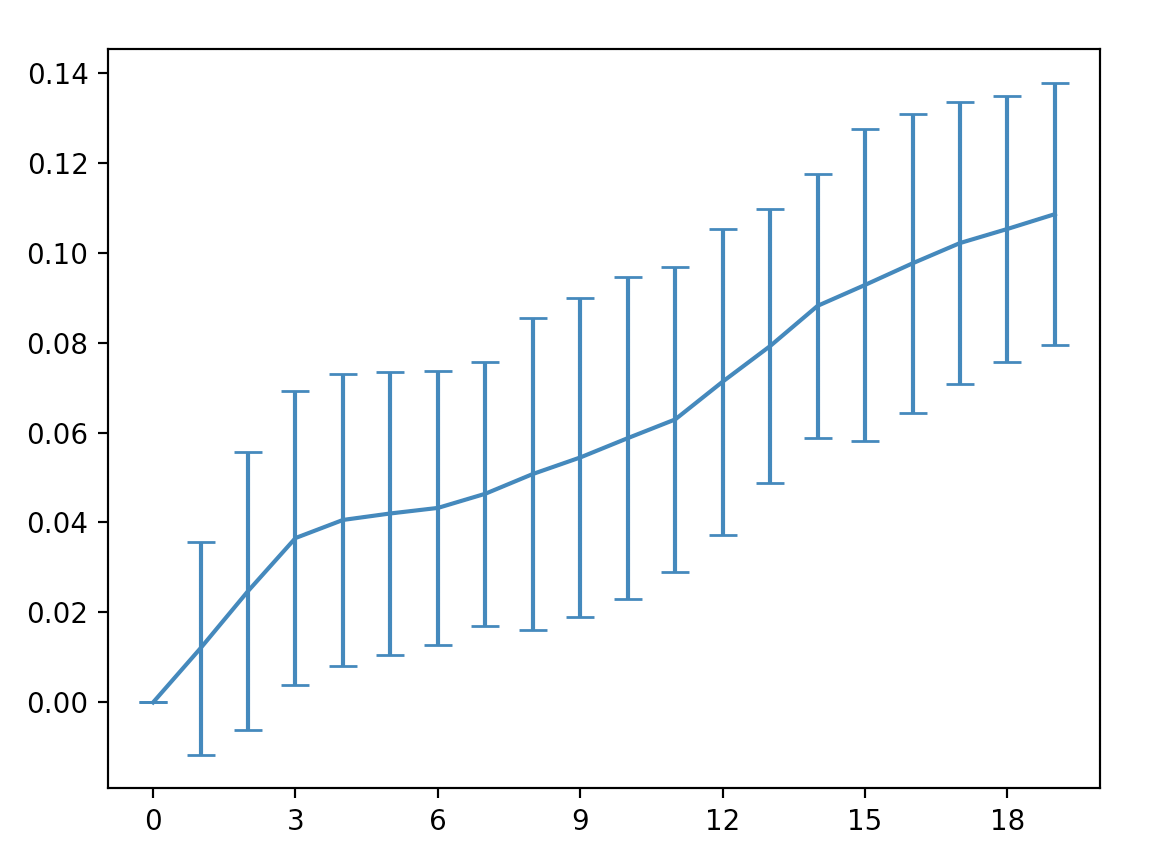} 
\caption{QAOA energy differences in terms of walk-distance  \label{walk_step}}
\end{figure}

\section{Landscapes}

\begin{figure*}[h!]\
    \centering
    \begin{subfigure}[b]{0.35\textwidth}
        \centering
        \includegraphics[width=\textwidth]{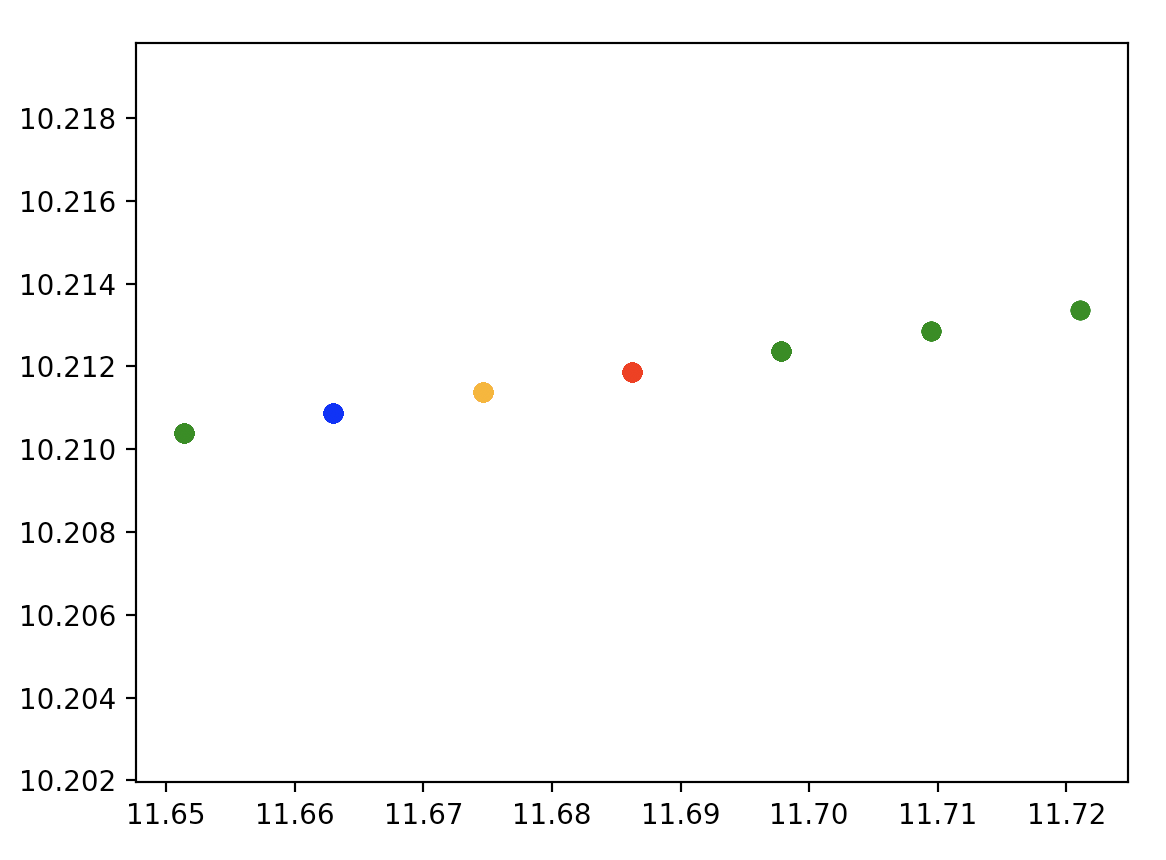}
        \caption{3-reg. graphs of size 14, $p=1$}
    \end{subfigure}
    \hspace{0.4in}
    \begin{subfigure}[b]{0.35\textwidth}
        \centering
        \includegraphics[width=\textwidth]{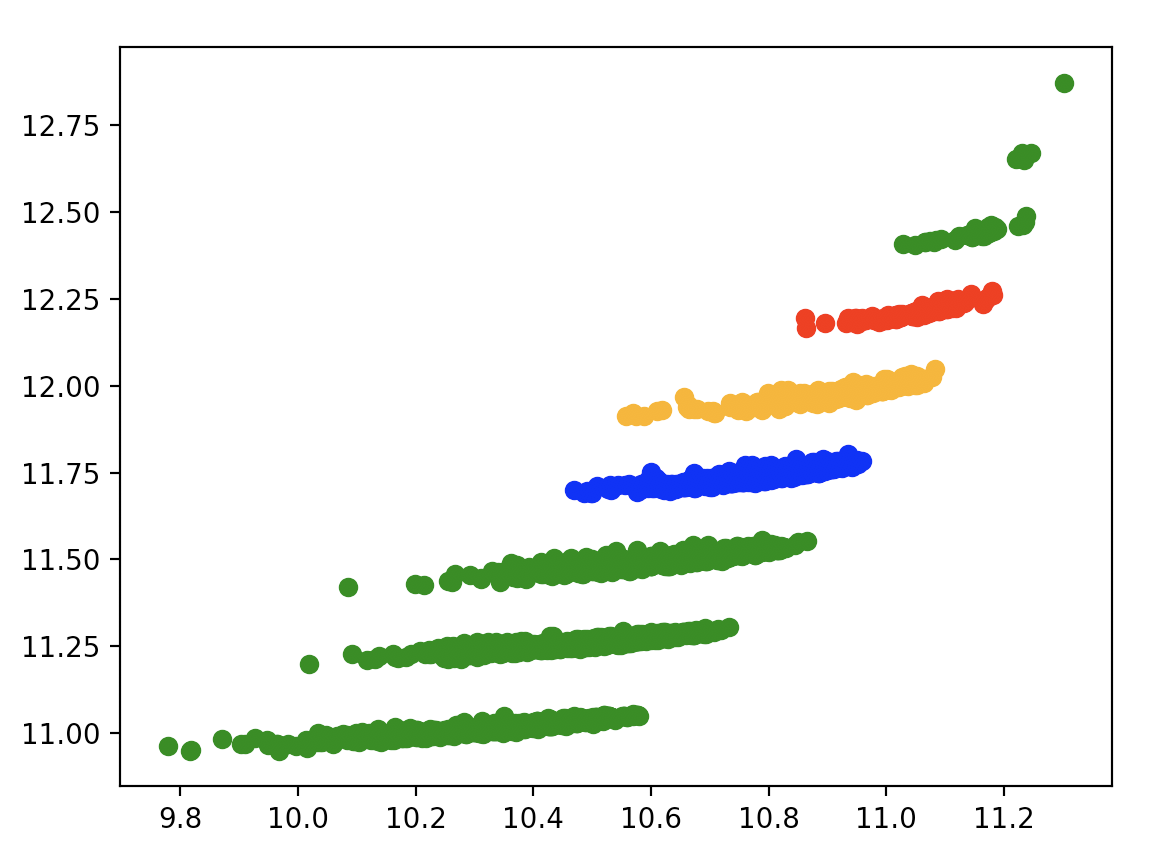}
        \caption{3-reg. graphs of size 16, $p=3$}
    \end{subfigure}%
   \\
    \begin{subfigure}[b]{0.35\textwidth}
        \centering
        \includegraphics[width=\textwidth]{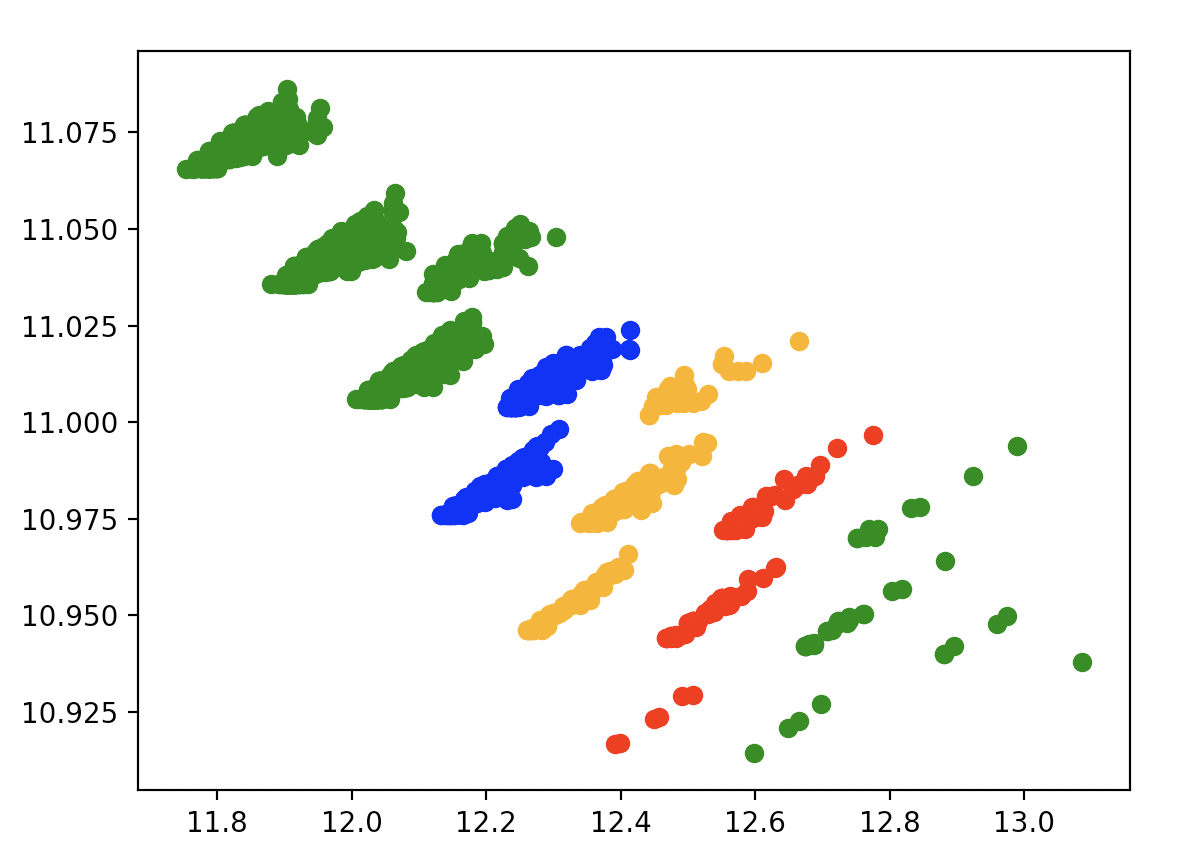}
        \caption{3-reg. graphs of size 16, $p=3$}
    \end{subfigure} 
    \hspace{0.4in}
    \begin{subfigure}[b]{0.35\textwidth}
        \centering
        \includegraphics[width=\textwidth]{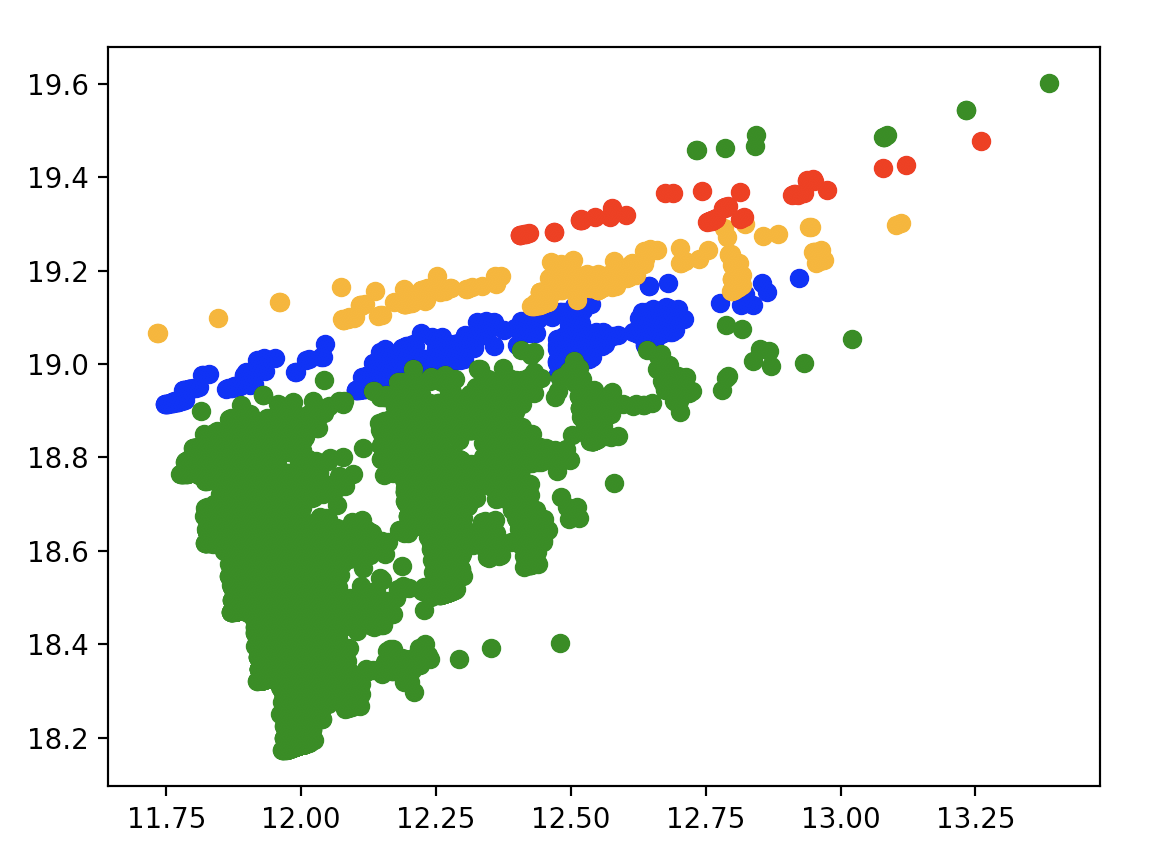}
        \caption{3-reg. graphs of size 18, $p=4$}
    \end{subfigure}%
    \caption{Different Landscapes of all 3-regular graphs on 16 and 18 nodes. \label{landscapes}
    Plots (b) and (c) are made from the same sets of graphs but using different angle sequences.}
\end{figure*}

A graph similarity measure is a metric on all or on a select set of graphs.
When this metric can be embedded into $\mathbb{R}^{2}$, we can draw a visual image of 
how the graphs in our set (like all 3-regular graphs of a given size) cluster according to the metric. We call such an embedding a {\em landscape},
which is not to be confused with energy landscapes. Our landscapes are {\em graph landscapes}.

Graph landscapes made from QAOA energies reveal intriguing properties of graphs.
Consider for instance 
 the ${\cal G}_{16,3}$, the set of all 3-regular graphs with 16 nodes. Let $p=3$, and fix $\gamma,\beta,\gamma',\beta' \in [0,2\pi]^{p}$.
 Define the map ${\cal G}_{16,3} \rightarrow \mathbb{R}^{2}$ by 
 \[
 G \rightarrow (E(G,\gamma,\beta), E(G,\gamma',\beta')
 \]
 Figure \ref{landscapes} (b) and (c) show two such landscapes made from different angle sequences.
 The clusters are clearly visible in both. Closer examination has revealed that the dominant clusters are formed by graphs with
 the same number of triangles. We have colored some subsets that contain graphs with the same number of triangles with different colors.
 Each color except the green 
 corresponds to a fixed number of triangles. 
 
 The dependence on the number of triangles is even more revealing when $p=1$.  Figure \ref{landscapes} (a)
 shows such a landscape of all three regular graphs with 14 nodes.
 The clusters shrink to a single point, which hints at Theorem \ref{triangle}, and in fact this is how we have discovered it.
 The same structure is not true for 4-regular graphs any more, and we need more graph invariants to explain the landscape. 
 
The clustering phenomenon 
\cite{shaydulin2019evaluating, shaydulin2019multistart}.

\newpage

\section{What are the least distinguishable pairs of graphs?} \label{least}

\begin{figure}[H]
\centering
\begin{tabular}{ccccccc}
\includegraphics[width=0.15\textwidth]{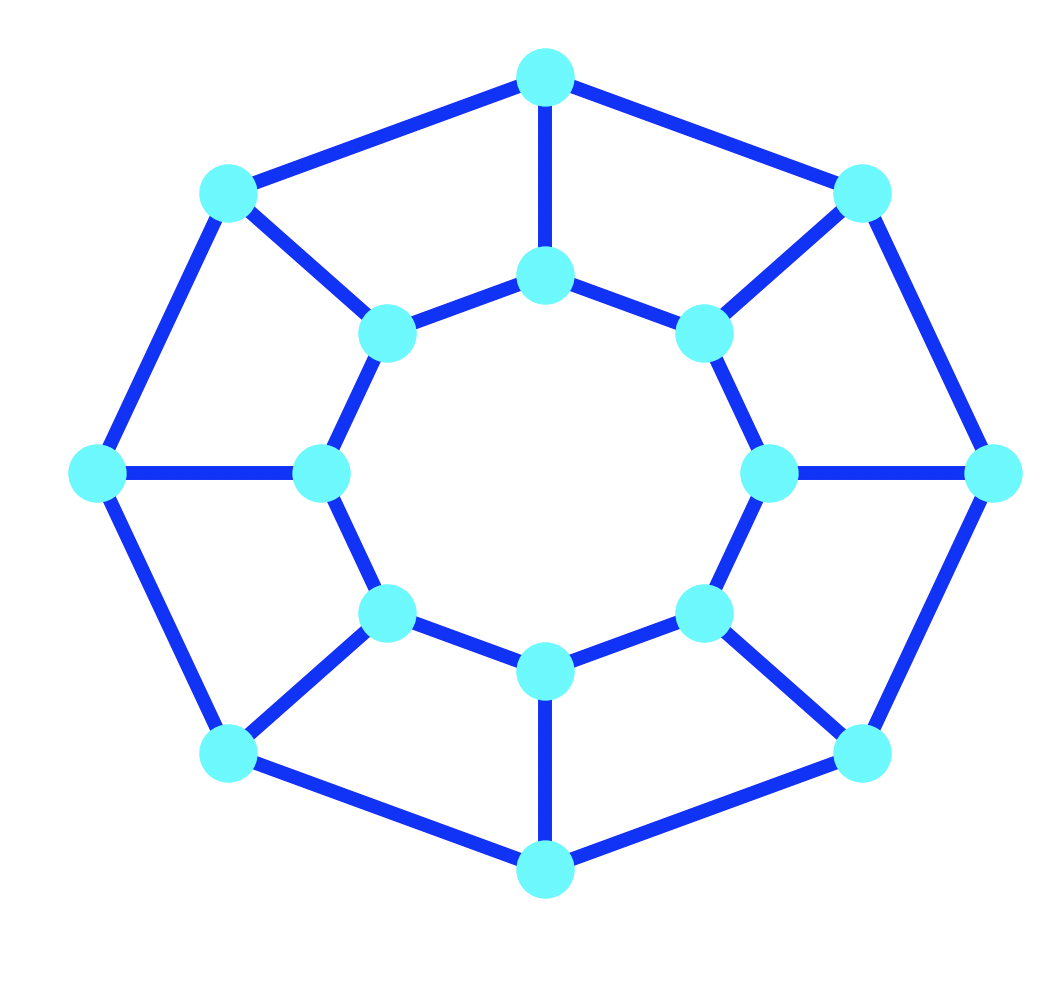} & \hspace{0.05in}  & \includegraphics[width=0.15\textwidth]{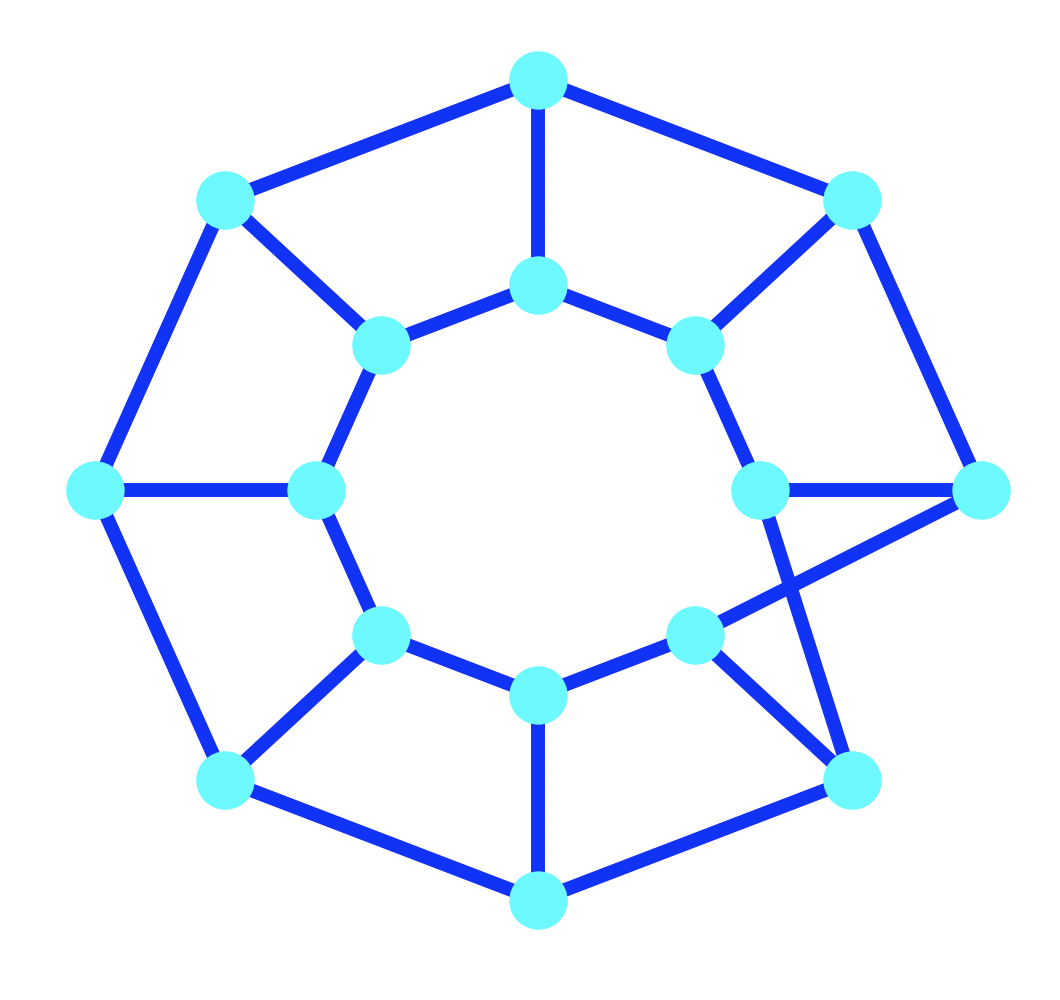}    & \hspace{0.2in}   &
\includegraphics[width=0.2\textwidth]{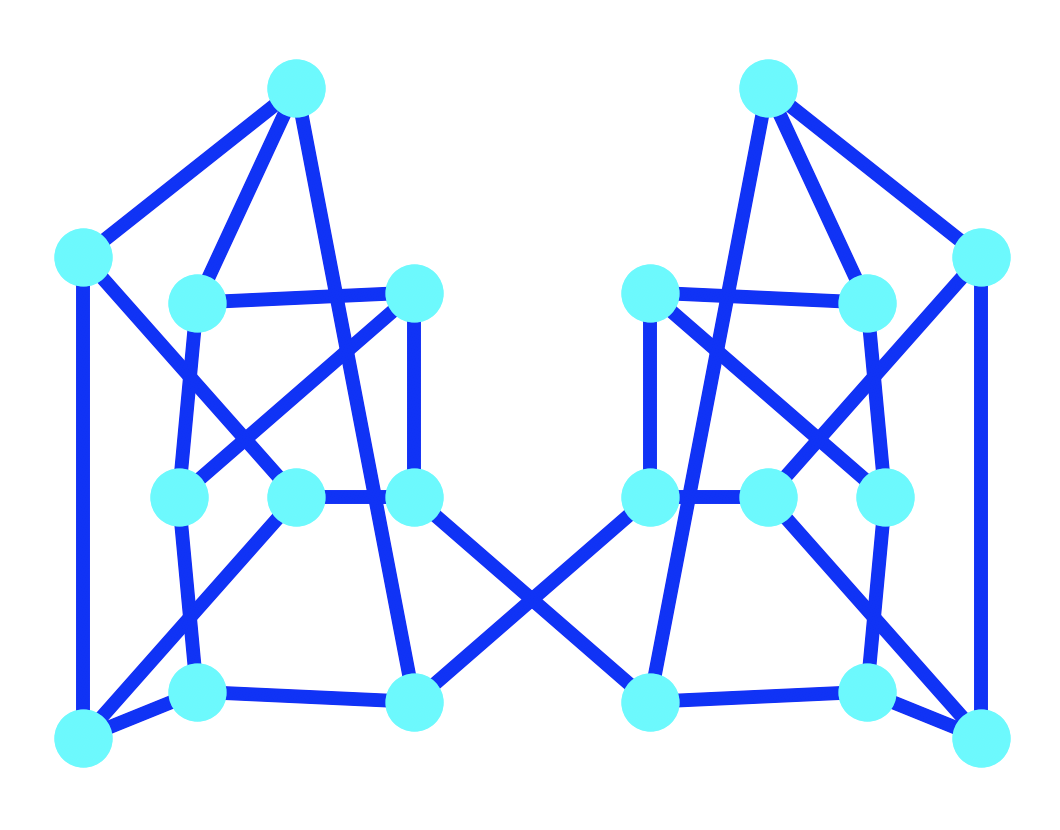} & \hspace{0.05in}  & \includegraphics[width=0.2\textwidth]{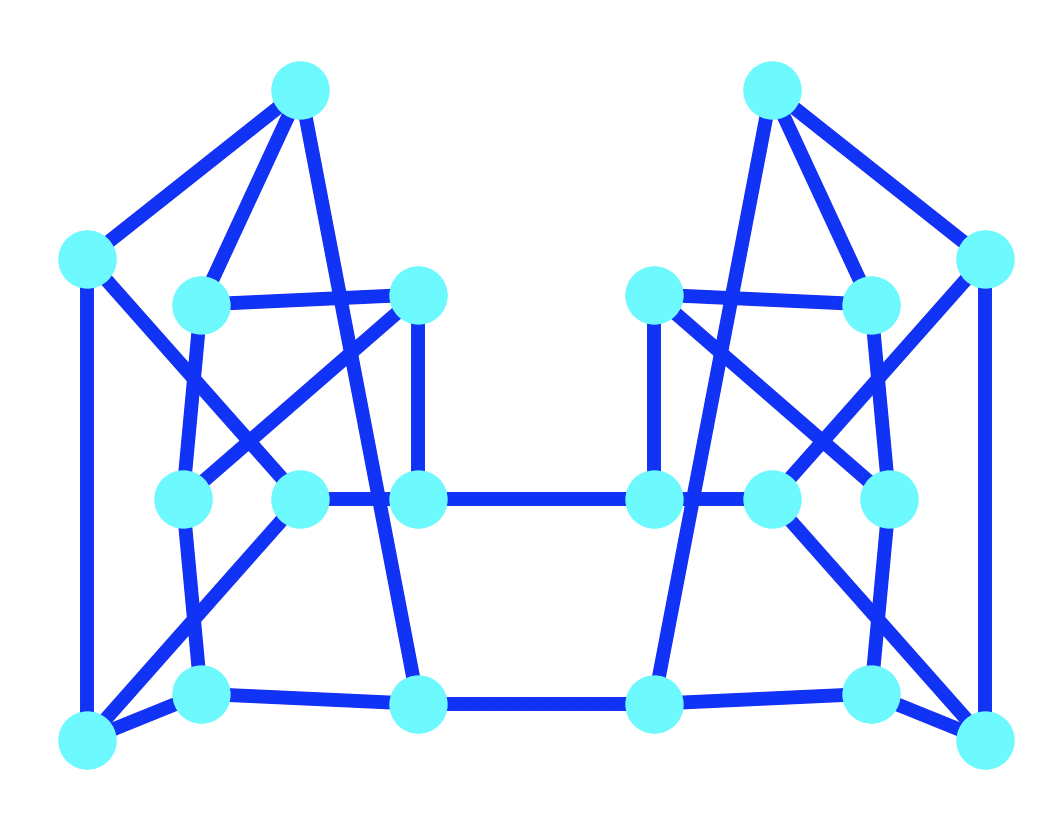}  \\
$CL(8)$ &  & $ML(8)$ & & $M_{1}$ & & $M_{2}$
\end{tabular}
\caption{Two hard-to-distingush pairs \label{cm}}
\end{figure}

Heuristic algorithms for the graph 
isomorphism problem are often tried on easily confusable 
graph pairs \cite{DBLP:journals/corr/NeuenS17a}. According to Laci Babai, only the most naive algorithms fall 
victim to strongly regular graphs with the same parameters or to iso-spectral pairs.
The Cai-Furer-Immerman \cite{DBLP:journals/combinatorica/CaiFI92} and Miyazaki graph pairs \cite{miyazaki, tener, neuen2018exponential} are a more serious challenge \cite{personal}.

For different GI heuristics different graph pairs are hard to distinguish \cite{PhysRevA.100.052317}. 
We have tested our algorithm on Miyazaki-inspired small examples, $M_{1}$ and $M_{2}$ on 20 nodes (Figure \ref{cm}, right two), each level 
from one to ten with 5000 random probes.
For our $M_{1}, M_{2}$ pair there was no gap up-to level 2. The separation has occurred at level 3:

\begin{figure}[H]
\begin{center}
\begin{tabular}{cccccccccc}
 level ($p$) & 2 & 3 & 4 & 5 & 6 & 7 & 8 & 9 & 10  \\\hline
gap   & 0  & 0.0003 & 0.0012 & 0.0026 & 0.0037 & 0.0054 & 0.0069 & 0.0083 & 0.0094 \\
squared & 0 & 1.7e-06 & 1.2e-05 & 3.4e-05 & 5.4e-05 & 9.4e-05 & 1.3e-04 & 1.7e-04 & 2.0e-04
\end{tabular}
\end{center}
\caption{$M_{1} =$ Miyazaki(20,1),  $M_{2} =$ Miyazaki(20,2): Expectation of gaps and squared gaps \label{cm}}
\end{figure}
 
While plotting the landscape for all connected 3-regular graphs with 16 nodes for levels one, two and three
we have noted that among all 4060 non-isomorphic three regular graphs, when reaching level three,
only a single pair of graphs remained that always had the same QAOA energies.
We have discovered that these two graphs were the Circular Lattice ($CL(8)$) and the Moebius Lattice ($ML(8)$)  (Figure \ref{cm}, left two).
This can be explained by that the $\lceil n/2\rceil - 1 $ edge-neighborhoods 
of $CL(n)$ and $ML(n)$ are the same. Following this clue we have examined the average gap from $n=3$ to $11$ 
at level 60 (maximal gap seems to be reached at levels $2n$ already, so we do not expect better separation for $p > 60$ and $n \le 11$).
and we have found 
the sequence exponentially decreasing.

\begin{figure}[H]
\begin{center}
\begin{tabular}{ccccccccccc}
size parameter ($n$) & 3 & 4 & 5 & 6 & 7 & 8 & 9 & 10 & 11 \\\hline
 $\Delta_{60}(CL, ML)$  & 1.07 & 0.81 &  0.47 & 0.283 & 0.182 & 0.1025 & 0.062 & 0.0337 & 0.015 \\
 ratio to the previous & N/A & 0.757 & 0.58 & 0.602 & 0.638 & 0.563 & 0.604 & 0.544 & 0.445
\end{tabular}
\end{center}
\caption{We can observe exponentially decreasing energy gaps for the $CL(n)$ and $ML(n)$ pair.  \label{ladder}}
\end{figure}

\medskip

This table gives us the strongest evidence so far that a QAOA-based graph isomorphism tester may not work.

\section{A Decoupling Phenomenon}\label{decoupling}

\begin{figure}[H]
\centering
\begin{tabular}{ccc}
\includegraphics[width=0.4\textwidth]{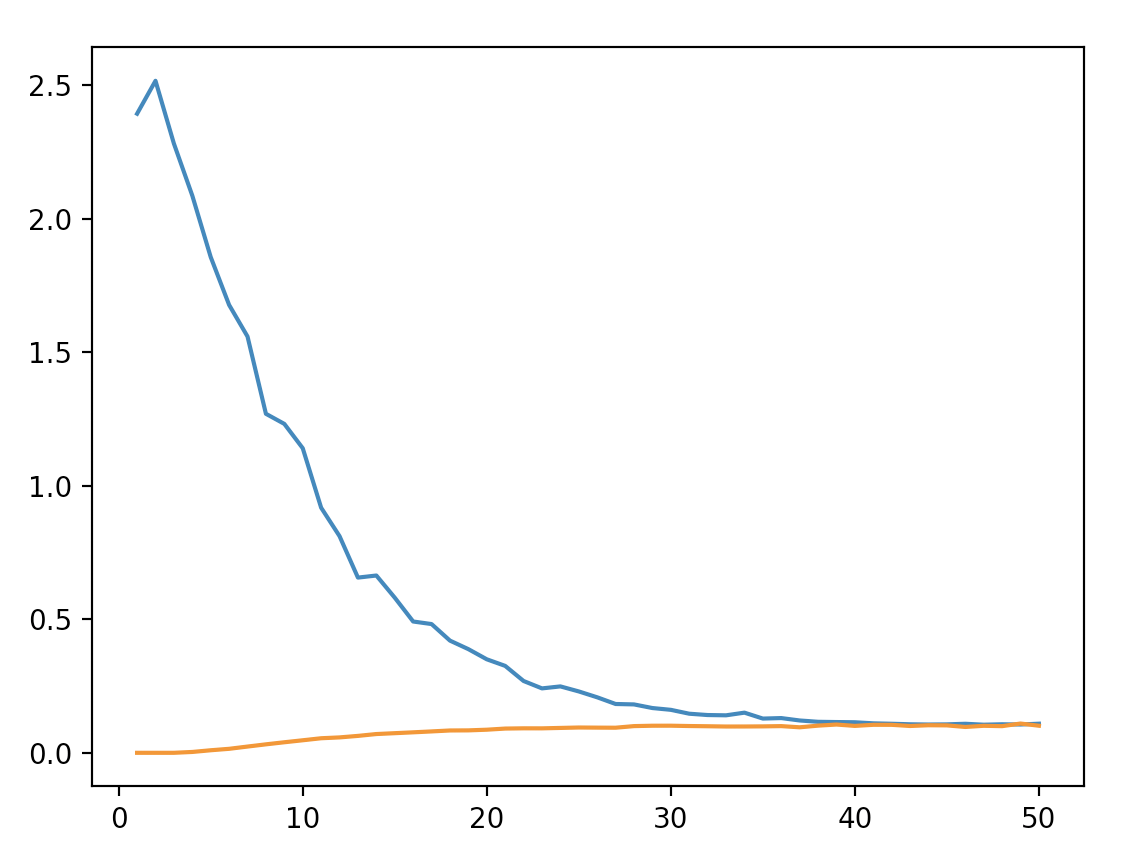} & $\;\;\;\;\;\;\;\;\;\;\;$ & \includegraphics[width=0.4\textwidth]{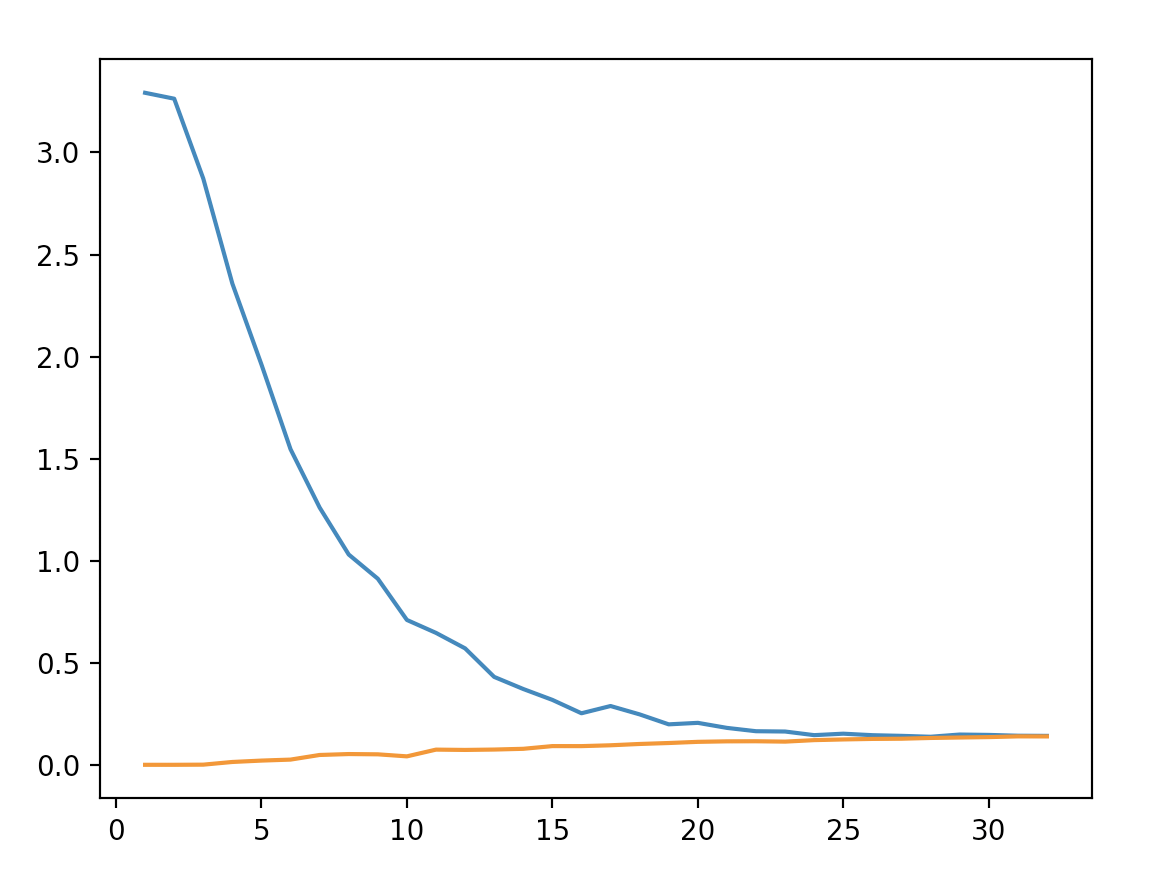} \\
$\Delta$(Circular,  Moebius) (16 nodes) & & $\Delta$(Praust1, Praust2) (20 nodes) \\
blue = uncorrelated; orange = identical & & blue = uncorrelated; orange = identical  \\
\includegraphics[width=0.4\textwidth]{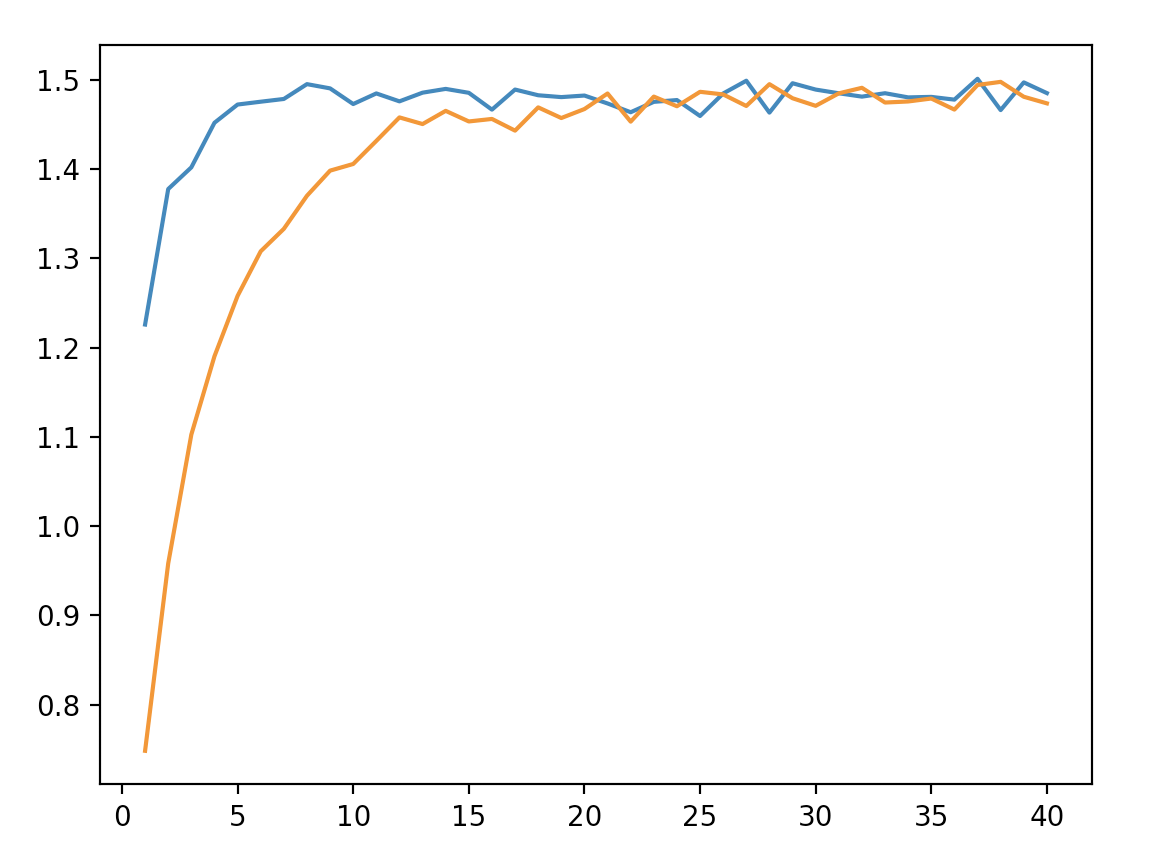} & & \includegraphics[width=0.4\textwidth]{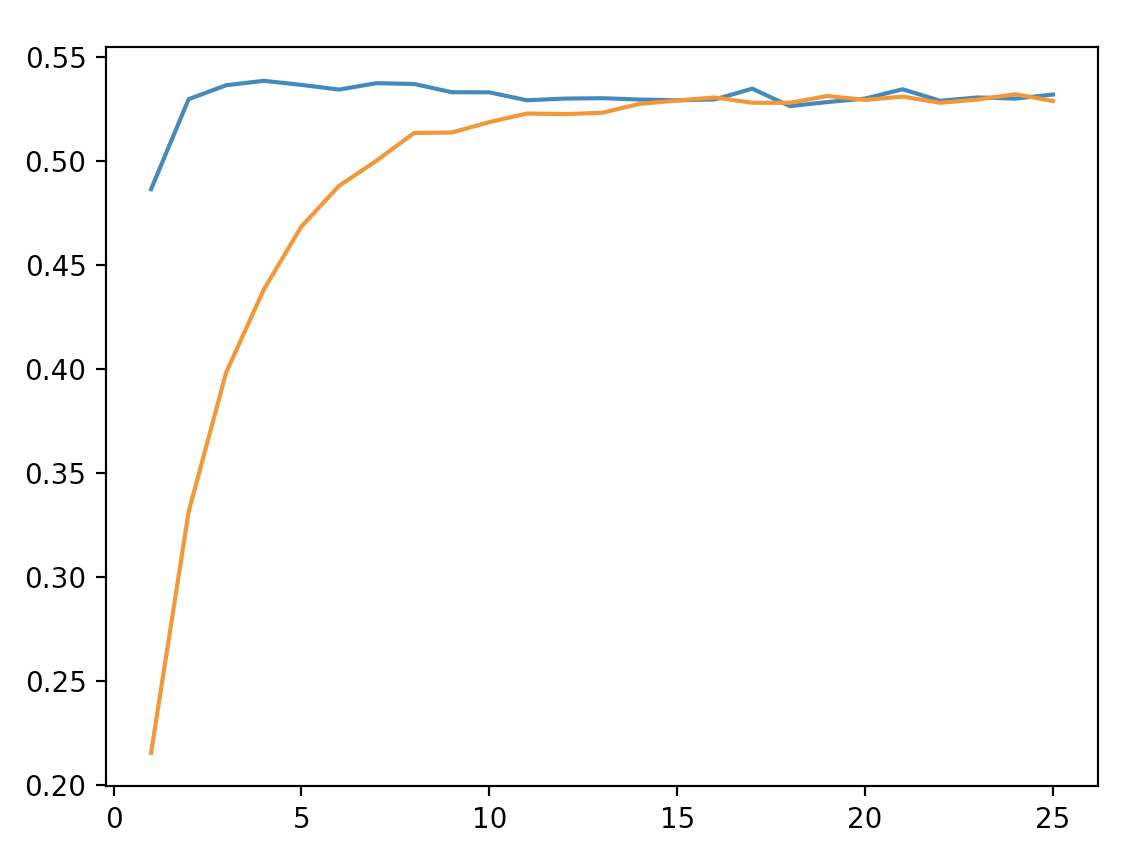} \\
$\Delta$(6-cycle,  $\triangle\triangle$) (6 nodes) & & $\Delta$(4-path, 3-star) (4 nodes) \\
blue = uncorrelated; orange = identical & & blue = uncorrelated; orange = identical 
\end{tabular}
\caption{Average energy gaps for un-correlated and identical angle sequences as $p$ increases.
For small $p$ un-correlated angle sequences yield much larger average gap $\Delta$, but the advantage 
disappears as $p$ gets larger. The  jigger in the curves is due to the relatively small sample size. \label{decoup1}}
\end{figure}

We have arrived at one of the most surprising findings of this research, which we have only experimentally asserted. 
Let $G_{1}$ and $G_{2}$ be two non-isomorphic graphs. We have considered the expected
energy difference for randomly chosen $(\gamma, \beta)$ in $[0,2\pi]^{2p}$:
\[
\Delta(G_{1}, G_{2}, p) = \E |E(G_{1},\gamma, \beta) - E(G_{2},\gamma, \beta)|
\]
where $p$ was $1,2,3,\ldots$. Then we have computed similar expectations, where we randomly 
and {\em independently} picked $(\gamma, \beta)$ and $(\gamma', \beta')$ in $[0,2\pi]^{2p}$:
\[
\Delta^{\rm indep}_{p}(G_{1}, G_{2},p)  = \E |E(G_{1},\gamma, \beta) - E(G_{2},\gamma', \beta')|
\]
What we have found is
\[
\lim_{p\rightarrow\infty} |\Delta^{\rm indep}(G_{1}, G_{2}, p)  - \Delta(G_{1}, G_{2}, p) | \;\longrightarrow 0
\]
In other words, the QAOA energy of $G_{1}$ decouples from that of $G_{2}$ as we probe them on longer and longer identical
angle sequences (of course only when graphs $G_{1}$ and $G_{2}$ are not isomorphic). 

The expression $\Delta^{\rm indep}_{p}(G_{1}, G_{2}, p)$ is uniquely determined by the 
distributions of $E(G_{1},\gamma, \beta)$ and $E(G_{1},\gamma, \beta)$ alone.
But the expectation, $ \square^{\rm indep}$, of the {\em squared} differences is easier to compute
from these distributions than $\Delta^{\rm indep}$:
\begin{eqnarray}\label{sqeq}
 \square^{\rm indep}(G_{1}, G_{2},p) = \E |E(G_{1},\gamma, \beta) - E(G_{2},\gamma', \beta')|^{2} & = & A + B - 2CD 
 \end{eqnarray}
 where $A, B, C$ and $D$ are the following expectation values for uniformly random $(\gamma, \beta)$ in $[0,2\pi]^{2p}$:
 \begin{eqnarray*}
A  & = & \E E(G_{1},\gamma, \beta)^{2} \\
B & = & \E E(G_{2},\gamma, \beta)^{2} \\
C & = & \E  E(G_{1},\gamma, \beta)  \\
D & = & \E E(G_{2},\gamma, \beta)
\end{eqnarray*}

In the last (third) part of the paper  $ \square^{\rm indep}$ is what we are going to 
analyze by finding ways to compute $A$, $B$, $C$ and $D$. Equation (\ref{sqeq}) turns out to have an interesting consequence 
if $\square^{\rm indep}(G_{1}, G_{2},p)$ is ``exponentially'' small (like $10^{-30}$). Since
\[
A + B - 2CD \;  \ge \; A-C^{2} + B-D^{2} \; = \; {\bf Var} \left( E(G_{1},p)\right) + {\bf Var} \left( E(G_{2},p) \right)
\]
the energy landscapes of both $G_{1}$ and $G_{2}$ at level $p$ are extremely flat.

\section{Efficient graph isomorphism with QAOA? -- an assessment}

Even if QAOA energies of non-isomorphic graphs differ for randomly chosen degree sequences and for large enough $p$ as Conjecture \ref{iso-conj} hypothesizes,
we are not guaranteed to have a polynomial time quantum algorithm for testing graph isomorphism, because 
the energy gap may be too small or $p$ may be too large. In Section \ref{least} we have unfortunately seen indications for the former. 
Below we formulate a sufficient condition
for efficient isomorphism testing (under the usual circuit model, which is also the model for Shor's algorithm).

\begin{theorem}\label{polytime}
Assume that there are polynomials $P(n)$ and $Q(n)$ such that for any two graphs $G_{1}$ and $G_{2}$ on 
$n$ nodes there exists some $p\le P(n)$ such that 
\[
\Delta(G_{1},G_{2},p) \ge {1\over Q(n)}
\]
Then our main algorithm in Section \ref{algsection} can be turned into a quantum-polynomial time
graph isomorphism solver.
\end{theorem}
\begin{proof}
First, it is a minor issue that the theorem does not assume the knowledge of $p$, since we can just try all $1\le p\le P(n)$.
The level-$p$ QAOA circuits for $G_{1}$ and $G_{2}$ have at most $P(n) (n + \max\{ |E(G_{1})| , |E(G_{2})|\} )$ gates.
With error correction, each gate as well as the input preparation and the final measurement can be implemented with $1 - 1/R(n)$ fidelity, where we choose
\[
R(n) = 100\, n^{2} Q(n) (P(n) + 2) (n + \max\{ |E(G_{1})| , |E(G_{2})|\} )
\]
Then, the fidelity of the entire circuit is not less than $1 - 1/(100 n^{2} Q(n))$.
The QAOA energies of $G_{1}$ and $G_{2}$ are upper bounded by $n^{2}$, so 
by Markov's inequality the probability that for a random degree sequence 
the QAOA energy difference of $G_{1}$ and $G_{2}$ is at least $1/(2Q(n))$ is at least ${1\over 2 n^{2} Q(n)}$.
The $1 - 1/(100 n^{2} Q(n))$ fidelity of the output enables us to compute the QAOA energy 
of both $G_{1}$ and $G_{2}$ for a given degree sequence
with at least $1/(50 Q(n))$ precision in $1 - \epsilon$ fraction of the time,
where $\epsilon$ can be made any inverse polynomial
(by sampling the energy sufficiently many times for the degree sequence in question).

Set $\epsilon = {1\over 2 K^2 n^{2} Q(n)}$ and run 
the distinguisher on $K n^{2} Q(n)$ random degree sequences.
Here $K$ is a user-defined number, where the user wants to achieve $O(1/K)$ error rate for the algorithm.
Analyse the cases:

\smallskip

\noindent{Case 1:} $G_{1}$ and $G_{2}$ are isomorphic.
In this case the true energy difference is zero for all degree sequences. By the in union bound the probability that
we turn up something larger than $1/(4Q(n))$ in $K n^{2} Q(n)$ rounds (the rounds correspond to different random degree sequences) is at most $K n^{2} Q(n)\epsilon \in O(1/K)$
for $G_{1}$ and the same for $G_{2}$, so the probability of failure is at most $1/K$.

\medskip

\noindent{Case 2:} $G_{1}$ and $G_{2}$ are not isomorphic.
In this case the $K n^{2} Q(n)$ probes must hit the set of angle sequences where the energy difference is at least $1/(2Q(n))$
at least once (but typically $\Omega(K)$ times) with probability $1-2^{-O(K)}$. When this happens, the algorithm may fail to output an energy difference 
$1/(4Q(n))$ or greater with probability at most $O(1/K)$. So the error probability in this case is again at most $O(1/K)$.
\end{proof}

A similar argument lets us replace $\Delta$ with $\square$ (the expectation of squared energy differences).
Next we look at the indications and counter-indications that the conditions of Theorems \ref{polytime} holds.

\subsection{Indications and counter indications}

Our original hope of building an efficient QAOA-based graph isomorphism tester 
was based upon that we could separate all non-isomorphic pairs of graphs that our 
QAOA simulator could handle, and even 
entire classes, like all 3-regular graphs up to 18 nodes and strongly regular classes up to 26 nodes.
Further, the separation usually happened already at a very low QAOA level:

\begin{center}
\begin{tabular}{|c|c|} \hline
Class or Pair of Graphs & QAOA Depth Giving Full Separation  \\\hline\hline\\\\[-3.5\medskipamount]
Myazaki I and II, 20 nodes & 4  \\\hline\\\\[-3.5\medskipamount]
Praust I and II, 20 nodes  & 4  \\\hline\\\\[-3.5\medskipamount]
All 4060 non-iso 3-regular graphs &  4  \\\hline\\\\[-3.5\medskipamount]
All 41301 non-iso 3-regular graphs on 18 nodes &  4  \\\hline\\[-1.5\medskipamount]
All 10 non-iso graphs in the SRG 26,10,3,4 family &  3  \\\hline
\end{tabular}
\end{center}

Another indication that the number of levels will not be an issue, is that
$\Delta(G_{1},G_{2},p)$ very quickly converges in terms of $p$, and even quicker becomes convex
in all examples we have looked at.
Therefore, we do not anticipate any problem with Conjecture \ref{iso-conj}, and we also believe 
that the number of levels where the separation happens is bounded
by a polynomial. If the conjecture is proven, it would be an analogue of 
the photon state -based tester result in \cite{bradler2018graph}.

The issue is the energy gap. For a pair of graphs let us set $p$ large enough that the expected gap 
almost reaches the plateau (our experience is that the gap stabilizes as $p$ grows, and that a polynomial $p$ is sufficient). We consider families of pairs of graphs.
Upon recognizing that the (Circular Ladder, Moebius Ladder)$_{n}$ family can be a counter example, we started to 
work on proving that the (best) gap is exponentially shrinking. We suspected that already the 
much simpler cycle sequence ($C_{2n}$, $C_{n} + C_{n}$)$_{n}$ shows the behavior.
That has turned out to be false -- somewhat of a positive sign.

The uncoupling phenomenon has raised hopes once again:
exponentially small un-coupled gap between two graphs would imply that the 
QAOA energy is very close to constant for both graphs. 
Surprisingly, however, our numerics strongly suggests that  
large circular ladders have indeed extremely flat energy landscapes.

\subsection{A no-go theorem for hyper-graphs}

Undoubtedly the most concrete evidence
against a polynomial (or even quasi-polynomial) time QAOA-based graph isomorphism tester is that the energy gaps between 
the Ladder and the corresponding Moebius Ladder graphs seem to be exponentially shrinking. Not only that, but
the shrinkage factor seems to even grow slightly (Figure \ref{ladder}). This is not the only counter-indication.
An abstract look at the problem yields that certain proof methods must fail.

MAXCUT is just one of the many problems to which QAOA applies.
Encouraged by the distinguishing power of QAOA on small MAXCUT instances we might try it on
small combinatorial structures other than graphs. 

\medskip

\noindent{\bf Weighted Constraint Satisfaction problems (WCSP).} A WCSP instance on $n$ bits is a set 
$\{C_{j}\mid 1\le j \le m\}$ of $\ell$-local constraints, i.e. each $C_{j}$ 
is a function $\{0,1\}^{\ell}\rightarrow \mathbb{R}$ that depends on $\ell$ of the $n$ bits.
The task is again to find an $n$-bit assignment that minimizes $\sum_{j = 1}^{m} C_{j}$.

The QAOA circuits are analogous to those made for MAXCUT, except instead of two-qubit gates we have 
$\ell$-qubit gates $e^{- i \gamma C_{j}}$.
Unfortunately, there are two 6-bit non-isomorphic WCSPs
that for all degree sequences give the exact same energy values. This is 
unlike the behavior of MAXCUT QAOA, where we have not found two non-isomorphic graphs so far with the exact same energy values for all $p$.

The counter-example is a pair of systems, we call them $H_{3,3}$ and $H_{6}$, each with $\ell=n=6$,
and each with a single constraint. These instances are not even ``weighted'' in the sense that 
their constraints are either satisfied by a $\sigma \in \{0,1\}^{6}$ assignment, i.e. return 1 or are not, i.e. return 0.

\bigskip

\begin{figure}[H]
\centering
\begin{tabular}{ccc}
$\begin{array}{c}
110000 \\
011000 \\
101000 \\
000110 \\
000011 \\
000101
\end{array}$
 & \hspace{1in} &
$ \begin{array}{c}
110000 \\
011000 \\
001100 \\
000110 \\
000011 \\
100001
\end{array}$ \\[4pt]
 $H_{3,3}$  & &  $H_{6}$ 
\end{tabular}
\caption{Assignments satisfying the (single) constraints of $H_{3,3}$ and $H_{6}$, respectively}
\end{figure}

\bigskip

\noindent The single constraint of $H_{3,3}$ is satisfied by those 6-tuples that have weight 2, and
each satisfying assignment is the indicator function of an edge in a union of two disjoint triangles. 
$H_{6}$ is similar, but the satisfying assignments are indicator functions of edges of a six-cycle. These instances should not be 
confused with MAX-CUT graph instances. In particular, the energy of 
an assignment $z$ for $H_{3,3}$ and $H_{6}$ is always either zero or one depending on whether $z$ is among the satisfying assignments or not.
In contrast, the MAX-CUT energy for an assignment is the number of edges of the input graph that the assignment does not cut.

To compute the QAOA energies we use the Feynman path approach.
A Feynman path $z$ for the level $p$ QAOA is a sequence 
\[
z_{0},\ldots,z_{p}\;\;\;\;\;\;\;\;z_{j}\in \{0,1\}^{n}
\]
where $n$ is the number of bits of the instance. 
The path corresponds to an up-going input-output path of the QAOA circuit.
To calculate the amplitude associated with this path, let us notice that each $z_{j}$ picks up a phase 
$e^{-i \gamma_{j} C(z_{j} }$, where  $C(z_{j})$ is the energy associated with $z$ by the (diagonal) instance Hamiltonian,
and the $z_{j}\rightarrow z_{j+1}$ transition has amplitude
\[
 (-i \sin \beta_{j} )^{\delta(z_{j},z_{j+1})}  (\cos \beta_{j} )^{n - \delta(z_{j},z_{j+1})}  \;\;\;\;\;\;\; ( \delta(x,y) \; \mbox{is the Hamming distance between} \; x, y \in \{0,1\}^{n})
\]
Here we are concerned only with the instances $H_{3,3}$ and $H_{6}$, so $C(z)$ is either 0 or 1.
When computing the QAOA energy we sum up the amplitude squares for all the paths that terminate in any $z_{p}$ with $C(z_{p})=1$.
To compute the amplitude square, for every such $z_{p}$ we 
need to consider products of pairs of such paths, the second component conjugated.
Thus, if $C$ is any of $H_{3,3}$ or $H_{6}$:
\begin{equation}\label{ee}
\langle \gamma,\beta | C |  \gamma,\beta \rangle
= {1\over 2^{6}} \sum_{z,z': z_{p} = \; z'_{p}\; \wedge \; E(z_{p})=1} \theta^{\beta} (z,z') \prod_{j=1}^{p-1}  e^{-i \gamma_{j}(C(z_{j}) - C(z'_{j}))}
\end{equation}
where
\begin{eqnarray*}
\theta^{\beta} (z,z') & = & \prod_{j=1}^{p-1}  \;
 (-i \sin \beta_{j} )^{\delta(z_{j},z_{j+1}) }  \;
(\cos \beta_{j} )^{n -\delta(z_{j},z_{j+1})}  
\;
\prod_{j=1}^{p-1} \;
 (i \sin \beta_{j} )^{\delta(z'_{j},z'_{j+1})}  \;
(\cos \beta_{j} )^{n - \delta(z'_{j},z'_{j+1})}  
\end{eqnarray*}
What makes the QAOA energies the same for $C = H_{3,3}$ and $C = H_{6}$
is that one can find a one-one correspondence 
between their Feynman paths pairs $z,z': z_{p} = \; z'_{p}\; \wedge \; E(z_{p})=1$ with the same contribution. 

\medskip

To define such a correspondence, let $Z_{w} \subseteq \{0,1\}^{6}$ be the set of assignments with Hamming weight $w$.
Let $C$ be either $H_{3,3}$ or $H_{6}$.
For every $\epsilon_{1},\epsilon_{2}\in\{0,1\}$ and $0 \le w_{1}, w_{2}, d\le 6$ there turns out to be a unique number, 
 $S_{C}(\epsilon_{1},\epsilon_{2}, w_{1}, w_{2}, d)$, that tells how many ways one find $x_{2}\in Z_{w_{2}}$ 
 for an $x_{1}\in Z_{w_{1}}$ with $C(x_{1}) = \epsilon_{1}$, that is fixed in advance (it turns out, all choices of $x_{1}$
 give the same numbers), such that
 \begin{enumerate}
 \item $C(x_{2}) = \epsilon_{1}$
 \item $d(x_{1},x_{2}) = d$
 \end{enumerate}
 Not only that $S_{C}(\epsilon_{1},\epsilon_{2}, w_{1}, w_{2}, d)$ does not 
 depend on the choice of $x_{1}$, but it {\em does not depend 
 on whether $C$ is $H_{3,3}$ or $H_{6}$}. So we set $S=S_{C}$. Here we list some 
 of the values of $S$:
 \[
 \begin{array}{ccc}
 S(0,0, 0,0,0) &  = &   1  \\
 S(0,0, 0,0,1) &  = &   0  \\
 S(0,0, 0,2,2) &  = &   9  \\
 S(0,1, 0,2,2) &  = &   6  \\
 S(0,1, 1,2,1) &  = &   2  \\
 S(0,1, 1,2,3) &  = &   4  \\
 S(1,1, 2,2,2) &  = &   2  \\
   &  \vdots &     \\
 \end{array}
 \]
 We leave it to the reader to verify that the existence of such an $S$ is sufficient to match up the terms 
 in Equation (\ref{ee}) for $H_{3,3}$ and $H_{6}$
 (hint: parse $z_{0},\ldots,z_{p},z'_{p-1} \ldots z'_{0}$ from left to right, and find a match by iteratively matching the $z_{j}$s and $z'_{j}$s).

\newpage
\part{The QAOA Dynamics}
The research in part has grown out of the following problem:

\bigskip

\noindent {\tt Prove or disprove the observed exponential decrease of the average energy gap \\
(Table \ref{ladder}) 
for the  Circular vs. Moebius Ladder pairs.}

\bigskip
In Table \ref{ladder} $p=60$ was chosen, but in general we want to pick a $p$ which increases 
as the instances increase. According to our experiments, for any pair
of graphs the energy gap stabilizes rather quickly, and we have never seen any violation of 
\[
\lim_{p\rightarrow\infty}\; \Delta(G_{1},G_{2},p) \approx   \Delta(G_{1},G_{2}, 3n) \;\;\;\;\; {\rm where}\;\;\; n = |V(G_{1})| = |V(G_{2})| 
\]
Similar observation holds for the easier to handle $ \square(G_{1},G_{2},p)$ quantity.
Due to the hypothesized Equation (\ref{square}), estimating the latter reduces to estimating QAOA moments,
at least for large enough $p$. Let us elaborate on this. In 
the introduction we have informally defined the QAOA dynamics. We now define it formally.

\medskip

\noindent{\bf QAOA Dynamics} Let $G$ be a graph on $n$ nodes. We shall denote probability distributions on the set
${\cal S} = \{z\mid z \in \mathbb{C}^{2^{n}}, \; |z| = 1\}$ of $n$ qubit pure states, by upper case sigmas. 
We stick to the usual system of notations in probability theory and treat ${\cal S}$ just as a set.
Thus, for a probability measure $\Sigma$ on ${\cal S}$ the $d \Sigma(\omega)$ notation for $\omega\in {\cal S}$
returns the infinitesimal measure that $\Sigma$ assigns to $\omega$, that can be only used inside an integral.
The stochastic map ${\cal Z}$ that takes a distribution on 
${\cal S}$ to a new distribution on ${\cal S}$ when we add a new level of QAOA with random angles, is a random mixture of unitaries:

\smallskip
\mybox{lightgray}{
\begin{equation}\label{primal}
{\cal Z} =
\left\{  \, d\beta d\gamma \cdot \prod_{i \in V(G)} e^{-i\beta X_{i}} \prod_{\langle jk\rangle \in E(G)} e^{-i\gamma C_{\langle jk\rangle}} \right\}_{\beta,\gamma\in U [0,2\pi]^{2}}  
\end{equation}}

\medskip

\noindent In the dual form, let $\Sigma$ be a probability measure on ${\cal S}$. Then for $\omega\in {\cal S}$:

\medskip

\mybox{lightgray}{
\begin{equation}\label{dual}
d({\cal Z} (\Sigma))(\omega) = \int_{[0,2\pi]^{2}} \left(
 \; d\Sigma\left(\prod_{\langle jk\rangle \in E(G)} e^{i\gamma C_{\langle jk\rangle}}  \prod_{i \in V(G)} e^{i\beta X_{i}} \; \omega\right) \right)
 \; d\beta \, d\gamma 
\end{equation}}

\medskip

$d({\cal Z} (\Sigma))$ uniquely defines ${\cal Z}(\Sigma)$ in the sense that it tells how to integrate over ${\cal Z} (\Sigma)$.
Let $\Sigma_{0}$ be the 
distribution, which is concentrated on $ |+\rangle^{\otimes n}$ (the zeroth layer QAOA).
We conjecture that the QAOA dynamics takes $\Sigma_{0}$ to a limiting distribution as $p$ tends to infinity:

\begin{conjecture}\label{conj}
The limiting distribution $\Sigma_{\infty} = \lim_{p\rightarrow\infty} {\cal Z}^{p} (\Sigma_{0})$ exists.
\end{conjecture}

\begin{remark} The weakest notion of limit that is still useful for us 
is via the {\em weak convergence} of measures. We require that for all bounded, continuous functions $F$ on the $2^{n}$ dimensional complex unit sphere
$ \lim_{p\rightarrow\infty} \; \int F \, d({\cal Z}^{p}\, \Sigma_{0})$ exists and equals to $\int F \, d(\Sigma_{\infty})$.
\end{remark}

We shall accept Conjecture \ref{conj} as true in the sequel.

\section{On von Neumann's Simplification of Statistical Ensembles}

A basic tenet of quantum information theory, that goes back to von Neumann, states that when we see a 
probability distribution $p_{1},\ldots,p_{k}$ on pure quantum states $|\psi_{1}\rangle,\ldots, |\psi_{k}\rangle$,
we can turn the object into the density matrix 
$\sum_{i=1}^{k} p_{i} |\psi_{i}\rangle \langle \psi_{i} |$ without losing information that matters.
This is true in most settings, but because in our study
statistical ensembles
\[
\{ p_{1}\cdot |\psi_{1}\rangle, \ldots,p_{k} \cdot |\psi_{k}\rangle\}
\]
arise differently from how they arise in physics, 
we cannot use the simplification rule of von Neumann without further ado.

Let us clearly state the dos and don'ts. Assume Alice has a machine ${\cal M}_{1}$ with a red button, which 
when Bob pushes, he uniformly gets either $|0\rangle$ or $|1\rangle$. Let she have 
another machine ${\cal M}_{2}$, which looks exactly like ${\cal M}_{1}$, but upon the push of the button Bob uniformly gets either $|+\rangle$ or $|-\rangle$.
Then there is no way Bob could distinguish between the two machines, however long he plays with them, since the
output of both are expressed with the same density matrix:
\[
\{ 0.5 \cdot |0\rangle , \; 0.5 \cdot |1\rangle\} \;\;\;\; \cong \;\;\;\; \{ 0.5 \cdot |+\rangle , \; 0.5 \cdot |- \rangle\}  \;\;\;\; \Longleftrightarrow \;\;\;\;
\left(
\begin{array}{cc}
0.5 & 0 \\
0 & 0.5
\end{array}
\right)
\]
No matter how many times Bob tries the buttons, no matter what 
measurements he performs on the output states,
he remains un-informed about whether he holds ${\cal M}_{1}$ or ${\cal M}_{2}$.

Let us now give Bob the supernatural power that when he gets a (pure) state $|\psi\rangle$,
he can compute the energy value $\langle \psi | {\cal H} |\psi\rangle$ associated with the Hamiltonian 
\begin{equation}\label{alice}
{\cal H} = 
\left(\begin{array}{cc}
1 & 0 \\
0 & 0
\end{array}\right)
\end{equation}
Then whether he holds ${\cal M}_{1}$ or ${\cal M}_{2}$ is no secret anymore. 
Even pushing the button only once, if the energy is either zero or one,
Bob holds the first machine, and if it is ${1\over 2}$ he holds the second.
If he pushes the button 6 times he will measure energy values
something like $1,0,0,1,1,0$ (${\cal M}_{1}$) versus ${1\over 2},{1\over 2},{1\over 2},{1\over 2},{1\over 2},{1\over 2}$
(${\cal M}_{2}$).

Alice can easily give Bob the power of computing $\langle \psi | {\cal H} |\psi\rangle$. She just installs a blue button on her 
quantum computer, which when Bob pushes, he gets the exact same pure state as in the previous run. Then by a few applications of the blue button
Bob can do a full state tomography on $\psi$, which in turn lets him estimate $\langle \psi | {\cal H} |\psi\rangle$. In contrast, in physics experiments the randomness that
selects a pure state from a set of pure states (thus yielding a mixed state) cannot be replicated and 
remembered: it is not under the control of the experimental device.

{\em One can use the density matrix formalism if and only if the randomness creating the mixture is inaccessible and not replicable.}

\medskip

\noindent{Note:} When showing this section to D. Ding, he has pointed out that S. Popescu in \cite{popescu}
describes a thought experiment leading to a very similar conclusion: any knowledge about the
preparation of a statistical ensemble renders the density matrix formalism insufficient.

\section{Higher Order Density Matrices}

Our goal will be to compute the {\em variance} of the QAOA energy values  
over random angle sequences when the graph and the level are fixed. 
If the angle sequence is also fixed, we have a pure QAOA state. 
By randomizing over angle sequences we have created a statistical ensemble out of these pure states.
This is a continuous ensemble, but for simplicity in this section we formulate everything in terms of discrete ensembles.
Our definitions easily  carry over to continuous ensembles.

Assume we have a statistical ensemble 
$\{p_{1}\cdot  |\psi_{1}\rangle, \ldots, p_{k}\cdot |\psi_{k}\rangle \}$ 
of pure states, and a Hamiltonian ${\cal H}$ on the 
same Hilbert space. What information about the ensemble is sufficient to compute 
the variance of the random variable $X$ defined by:
\[
\Pb (\, X  = \langle \psi_{i} | {\cal H} | \psi_{i} \rangle\, ) \; = \; p_{i} \;\;\;\;\;\; (1\le i \le k) \;\;\; ?
\]
The first thought would be that this should only depend on the density matrix, 
$\sigma = \sum_{i=1}^{k} p_{i}  |\psi_{i}\rangle\langle \psi_{i}|$. In particular, we may have the illusion that
what we are really computing is
\[
{\bf Var}_{\sigma}( {\cal H}) = {\bf Tr} ({\cal H}^{2}  \sigma) - {\bf Tr} ({\cal H}\sigma)^{2}
\]
But not! ${\bf Var}(X) =  {\bf Var}_{\sigma}( {\cal H})$ only if each $\psi_{i}$ is an eigenstate of ${\cal H}$. Otherwise
we have to subtract the non-negative quantity 
$\sum_{i=1}^{k} p_{i}  \left(\langle \psi_{i} | {\cal H}^{2}  | \psi_{i} \rangle  -  \langle \psi_{i} | {\cal H}\sigma)^{2}  | \psi_{i} \rangle\right)$
from ${\bf Var}_{\sigma}( {\cal H})$ to get ${\bf Var}(X)$, or do a direct calculation to arrive at:
\[
{\bf Var}(X) = \sum_{i=1}^{k} p_{i} \langle \psi_{i} | {\cal H} | \psi_{i} \rangle^{2} - {\bf Tr} ({\cal H}\sigma)^{2}
\]
We focus on the expression $\E (X^{2}) = \sum_{i=1}^{k} p_{i} \langle \psi_{i} | {\cal H} | \psi_{i} \rangle^{2}$,
since this is what requires explanation.
The example in the previous section can be pushed one step further to show that $\sigma$ is insufficient to express $\E (X^{2})$
(see example in the end of this section).
Thus we also define
\[
\mbox{Second Order Density Matrix:} \;\;\;\;\;\;\;  \sigma^{(2)}  \; = \; \sum_{i=1}^{k} \; p_{i} \; \times \; | \psi_{i} \rangle \langle \psi_{i} |\; \otimes \; | \psi_{i} \rangle \langle \psi_{i} |
\]
\noindent{\em Warning:} In general $\sigma^{(2)} \neq \sigma \otimes \sigma$, where $\sigma$ is the usual (first order) density matrix. The point 
of introducing  $\sigma^{(2)}$ is exactly that it carries information about the ensemble that $\sigma$ does not.

\begin{remark}
Higher order density matrices are a tool in quantum chemistry where we trace out all but $k$ electrons from an $N$-electron wave function \cite{Neuscamman}. 
In this article we use this term, but the content and context differ (there is a relation nevertheless). 
\end{remark}
We have
\begin{eqnarray*}
\E (X^{2}) & = & \sum_{i=1}^{k} p_{i} \langle \psi_{i} | {\cal H} | \psi_{i} \rangle^{2} =  {\bf Tr} ({\cal H}^{\otimes 2} \sigma^{(2)}) \\[5pt]
{\bf Var}(X) & = & {\bf Tr} ({\cal H}^{\otimes 2} \sigma^{(2)})  - {\bf Tr} ({\cal H}\sigma)^{2}
\end{eqnarray*}
In particular, the variance of the energy of a Hamiltonian with respect to pure states taken from a
statistical ensemble of pure states is entirely determined by 
the Hamiltonian and the latter's first and second order density matrices.

An ensemble $\{p_{1}\cdot  |\psi_{1}\rangle, \ldots, p_{k}\cdot |\psi_{k}\rangle \}$ is 
a classical probability distribution on the unit sphere $\mathbb{S}$ of a Hilbert space. Let us now also have a
statistical ensemble $\{q_{1}\cdot  U_{1}, \ldots, q_{l}\cdot U_{l} \}$ of unitary matrices, that
can be interpreted as a stochastic map (dynamics) from $\mathbb{S}$ to $\mathbb{S}$.
The first and second order 
density matrices $\sigma$ and $\sigma^{(2)}$ uniquely transform under such a map,
yielding $\xi$ and $\xi^{(2)}$, where
\begin{eqnarray}\label{firstorderdef}
\xi & =  & \sum_{i} q_{i} \; \times \; U_{i} \;\sigma\; U_{i}^{\dagger} \\\label{secondorderdef}
\xi^{(2)}  & =  & \sum_{i} q_{i} \; \times\; (U_{i}\otimes U_{i}) \;\sigma^{(2)}\; (U_{i}^{\dagger} \otimes U_{i}^{\dagger} )
\end{eqnarray}
The first equation is standard, and the second equation is easy to show. 

Maps as above are exactly the kind that take random level-p QAOA states into a random level-$(p+1)$
QAOA states. The only difference is that
the sums must be replaced with integrals because the distribution of the angles is continuous.

\bigskip

\mybox{lightgray}{\small 
\noindent{\bf Example.} 
Let $\Sigma_{1} = \{ 0.5 \cdot |0\rangle , \; 0.5 \cdot |1\rangle\}$  and 
$\Sigma_{1} = \{ 0.5 \cdot |+\rangle , \; 0.5 \cdot |- \rangle\}$ be statistical ensembles
on the one qubit complex unit sphere and ${\cal H}$ be a Hamiltonian as in (\ref{alice}).
Let  $X_{1} = \{ \langle \psi | {\cal H} | \psi \rangle\}_{\psi \in  \Sigma_{1} }$, $X_{2} = \{ \langle \phi | {\cal H} | \phi \rangle\}_{\phi \in  \Sigma_{2} }$
be two random variables. Then $\E (X^{2}_{1})  = 1/2$ and $\E (X^{2}_{2})  = 1/4$. We compute 
the second order density matrices,  $\sigma^{(2)}_{1}$ and $\sigma^{(2)}_{2}$ associated with $X_{1}$ and $X_{2}$:
\[
\sigma^{(2)}_{1} = 
\left(
\begin{array}{cccc}
0.5 & 0 & 0 & 0\\
0 & 0 & 0 & 0 \\
0 & 0 & 0 & 0 \\
0 & 0 & 0 & 0.5 \\
\end{array}
\right)
\;\;\;\;
\sigma^{(2)}_{2} = 
\left(
\begin{array}{cccc}
0.25 & 0 & 0 & 0.25 \\
0 & 0.25 & 0.25 & 0 \\
0 & 0.25 & 0.25 & 0 \\
0.25 & 0 & 0 & 0.25 \\
\end{array}
\right)
\]
Then ${\bf Tr} ({\cal H}^{\otimes 2} \sigma^{(2)}_{1} )  = 0.5$ and  ${\bf Tr} ({\cal H}^{\otimes 2} \sigma^{(2)}_{2} )  = 0.25$. 
}

\section{Density matrices under the QAOA dynamics}

We fix a graph $G$ and will study the distribution on $n$ qubit states that arise
by running the QAOA circuit with random, level $p$ degree sequences ($ {\cal Z}^{p} (\Sigma_{0})$, in our notations). 
%If we {\em were} to make the von Neumann simplification,
One can easily compute the density matrices corresponding to these statistical ensembles.
At level $p=0$ we have the density matrix $\sigma_{0} =( |+\rangle\langle + | )^{\otimes n}$. Each new level is an application of the
super-operator ${\bf Z}$ that acts on a density matrix $\sigma$ as:

\medskip

\mybox{lightgray}{
\[
{\bf Z}: \; \sigma \,\rightarrow\,  \int_{(\beta,\gamma)\in [0,2\pi]^{2}}  
\prod_{v \in V(G)} e^{-i\beta X_{v}} \prod_{\langle jk\rangle \in E(G)} e^{-i\gamma C_{\langle jk\rangle}}\;
\sigma \; 
\prod_{\langle jk\rangle \in E(G)} e^{i\gamma C_{\langle jk\rangle}}
\prod_{v \in V(G)} e^{i\beta X_{v}} 
\]}

\medskip

so the density matrix corresponding to a level $p$ random QAOA circuit is
\[
\sigma_{p} = {\bf Z}^{p} \left( \sigma_{0} \right)
\]
Accepting the message of the previous section however,
we will not try to represent the evolution of the random QAOA ensemble with density matrices alone.
However, if we keep track of the evolution of second order density matrices as well,
we already have enough information for what we need to compute.
Let $\sigma_{0}^{(2)} =  \sigma_{0}^{\otimes 2}$. This is the second order density matrix corresponding
to the statistical ensemble $\Sigma_{0}$, concentrated on $|+\rangle^{\otimes n}$.

By Equation (\ref{secondorderdef}) we can define the 
second order super-operator ${\bf Z}^{(2)}$ made from ${\cal Z}$ that takes second order density matrices to second order 
density matrices. For any second order density matrix  $\sigma^{(2)}$ we have

\medskip

\mybox{lightgray}{
{\small
\[
{\bf Z}^{(2)}: \; \sigma^{(2)} \,\rightarrow\,  \int_{(\beta,\gamma)\in [0,2\pi]^{2}}  
\prod_{v \in V(G)} e^{-i\beta X_{v}^{\otimes 2}} \prod_{\langle jk\rangle \in E(G)} e^{-i\gamma C_{\langle jk\rangle}^{\otimes 2}}\;
\sigma^{(2)} \; 
\prod_{\langle jk\rangle \in E(G)} e^{ i\gamma C_{\langle jk\rangle}^{\otimes 2}}
\prod_{v \in V(G)} e^{i\beta X_{v}^{\otimes 2}} 
\]}}

\medskip

\noindent The desired information about $\sigma_{\infty} = \lim_{p\rightarrow\infty} {\cal S}^{p} \, \sigma_{0}$ can be now 
obtained from 

\medskip

\mybox{lightgray}{
\begin{center}
\begin{tabular}{lll}
The first order density matrix of $\Sigma_{\infty}$ & is & $\sigma_{\infty}$ \\
The second order density matrix of $\Sigma_{\infty}$ & is & $\sigma_{\infty}^{(2)}$
\end{tabular}
\end{center}
\medskip
where
\begin{eqnarray*}
\sigma_{\infty} & = & \lim_{p\rightarrow\infty} {\bf Z}^{p} \, \sigma_{0} \\
\sigma_{\infty}^{(2)} & = & \lim_{p\rightarrow\infty} \left({\bf Z}^{(2)} \right)^{p} \, \sigma_{0}^{(2)} 
\end{eqnarray*}}

\section{Computing first order density matrices}

In this section we give a recipe for computing the evolution of density matrices under the QAOA dynamics. The integrals
${1\over 2\pi}\int_{0}^{2\pi} \sin^{a} x\cos^{b} x \, dx$ for $a,b \in \{0,1,2,\ldots\}$ will play a role in the calculations. Notice
that if either $a$ or $b$ is odd than the integral is zero. If $a=2r$ and $b=2s$ then
\begin{equation}\label{pascal}
A(r,s) \; = \; {1\over 2\pi}\int_{0}^{2\pi} \sin^{2r} x\cos^{2s} x \, dx \; = \; {(2r)! (2s)!\over 4^{r+s} (r+s)! r! s!}
\end{equation}

If we arrange the values of $A(r,s)$ in a triangle fashion, like

\medskip

\begin{center}
\begin{tabular}{>{$d=}l<{$\hspace{12pt}}*{13}{c}}
0 &&&&&&& $A(0,0)$ &&&&&&\\[5pt]
2 &&&&&& $A(0,1)$ && $A(1,0)$ &&&&&\\[5pt]
4 &&&&& $A(0,2)$ && $A(1,1)$ && $A(2,0)$ &&&&\\[5pt]
\end{tabular}
\end{center}

where $d=2p+2q$ is the total degree of the trigonometric polynomial in the integrand, we obtain:

\bigskip

\begin{center}
\begin{tabular}{>{$d=}l<{$\hspace{12pt}}*{13}{c}}
0 &&&&&&& $1$ &&&&&&\\[5pt]
2 &&&&&& ${1\over 2}$ && ${1\over 2}$ &&&&&\\[5pt]
4 &&&&& ${3\over 8}$ && ${1\over 8}$ && ${3\over 8}$ &&&&\\[5pt]
6 &&&& ${1\over 16}$ && ${5 \over 16}$ && ${5\over 16}$ && ${1\over 16}$ &&&\\[5pt]
8 &&& ${35\over 128}$ && ${5\over 128}$ && ${3\over 128}$ && ${5\over 128}$ && ${35\over 128}$ &&\\[5pt]
10 && ${63\over 256}$ && ${7\over 256}$ && ${3\over 256}$ && ${3\over 256}$ && ${7\over 256}$ && ${63\over 256}$ &\\[5pt]
12 & ${231\over 1024}$ && ${21\over 1024}$ && ${7\over 1024}$ && ${5\over 1024}$ && ${7\over 1024}$ && ${21\over 1024}$ && ${231\over 1024}$
\end{tabular}
\end{center}

\bigskip

Recall that for a graph $G$ on $n$ nodes we have $\sigma_{0} = (|+\rangle \langle + |)^{\otimes n}$, which is a $2^{n}$ by $2^{n}$ matrix with 
all entries ${1\over 2^{n}}$. The first observation is:

\begin{lemma}
Let $p>1$ and let
\[
C(z) = \sum_{\langle jk\rangle \in E(G)} C_{\langle jk\rangle}(z), \;\;\; C_{\langle jk\rangle}(z) = {1\over 2}(1 + \sigma_{j}^{z}\sigma_{k}^{z} ) \;\;\;\;
[\mbox{the number of edges not cut by $z$}]
\]
Then $\sigma_{p}$ depends only on those $(x',y')$ entries of $\sigma_{p-1}$ for which
 ${\rm C}(x') = {\rm C}(y') $.
\end{lemma}
\begin{proof} For $p\ge 1$ we have
\begin{equation}\label{feyn}
\sigma_{p} [x,y] \; = \;  \sum_{x'\in \{0,1\}^{n}}\; \sum_{y' \in \{0,1\}^{n}} {1\over 2\pi}\int_{0}^{2\pi} e^{i\gamma\cdot({\rm C}(x') - {\rm C}(y') } \, d\gamma \; \cdot \; 
X_{x\, y}^{x'y'}
 \; \cdot \;  \sigma_{p-1} [x',y'] 
\end{equation}
where
\begin{equation}\label{feyn2}
 X_{x\, y}^{x'y'} \; = \; {(-i)^{\delta(x,x') - \delta(y,y')} \over 2\pi}\int_{0}^{2\pi} (\sin\beta)^{\delta(x,x')+\delta(y,y')}  (\cos\beta)^{2n - \delta(x,x') - \delta(y,y')}  \, d\beta
\end{equation}
Notice that 
\begin{equation}\label{feyn3}
 {1\over 2\pi}\int_{0}^{2\pi} e^{i\gamma\cdot({\rm C}(x') - {\rm C}(y')} \, d\gamma \; = \;
 \left\{
 \begin{array}{ccc}
 1 & {\rm if} & {\rm C}(x') = {\rm C}(y') \\
 0 & {\rm if} & {\rm C}(x') \neq {\rm C}(y')
 \end{array}\right.
\end{equation}
so for any $x$ and $y$ those terms of the r.h.s. of (\ref{feyn}) that correspond to an $(x',y')$ with ${\rm C}(x') \neq {\rm C}(y')$ are zero.
\end{proof}
Expressions (\ref{feyn}), (\ref{feyn2}), (\ref{feyn3}) also give a recipe for computing $\sigma_{p}$ from $\sigma_{p-1}$:

\medskip

\mybox{lightgray}{
\[
\sigma_{p} [x,y] \; = \;  
\sum_{
\begin{array}{ccl}
x',y' & \in & \{0,1\}^{n}  \\
{\rm C}(x') & = &{\rm C}(y')  \\
\delta(x,x') & = & \delta(y,y')  \mod 2
\end{array}} 
(-1)^{\delta(x,x') - \delta(y,y')\over 2} \; \cdot \; 
A(\nu, n - \nu) \; \cdot \;  \sigma_{p-1} [x',y'] 
\]
where
\[
A(r,s) \; \mbox{is as in (\ref{pascal}), and} \;\;\;\; \nu \; = \; \nu_{x\, y}^{x'y'} \; = \; {\delta(x,x') + \delta(y,y')\over 2}
\]}

\section{Computing second order density matrices}\label{ssquare}

The computation of second order density matrices is very similar to that of the first order ones, and we only write down the expression.
Recall that $\sigma^{(2)}_{0} = ( |+\rangle\langle +|)^{\otimes2n}$. The recipe for computing $\sigma^{(2)}_{p}$ from $\sigma^{(2)}_{p-1}$ is:

\medskip

\mybox{lightgray}{
\[
\sigma^{(2)}_{p} [x_{(1)},x_{(2)}, y_{(1)}, y_{(2)}] \; = \;  
\sum_{
\begin{array}{ccl}
x'_{(1)},x'_{(2)}, y'_{(1)}, y'_{(2)} & \in & \{0,1\}^{n}  \\
{\rm C}(x'_{(1)}) +  {\rm C}(x'_{(2)}) & = & {\rm C}(y'_{(1)}) +  {\rm C}(y'_{(2)})  \\
2 \mid \delta(x_{(1)},x'_{(1)}) + \delta(x_{(2)},x'_{(2)}) & - & \delta(y_{(1)},y'_{(1)}) - \delta(y_{(2)},y'_{(2)})  
\end{array}} P \cdot Q \cdot R \]
where
\begin{eqnarray*}
P & = & (-1)^{\delta(x_{(1)},x'_{(1)}) + \delta(x_{(2)},x'_{(2)}) - \delta(y_{(1)},y'_{(1)}) - \delta(y_{(2)},y'_{(2)}) \over 2} \\
Q & = & A(\nu,n-\nu) \;\;\;\;\; \nu \; = \; {\delta(x_{(1)},x'_{(1)}) + \delta(x_{(2)},x'_{(2)}) + \delta(y_{(1)},y'_{(1)}) + \delta(y_{(2)},y'_{(2)}) \over 2}\\
R & = & \sigma^{(2)}_{p-1} [x'_{(1)},x'_{(2)}, y'_{(1)}, y'_{(2)}]  \\
\end{eqnarray*}}

\section{Conclusions}

This paper is a continuation of formula-driven QAOA research. We have proposed a few quantities to calculate, and could find some 
theoretical and intuitive tools for the calculations. Although our focus was graph structure discovery,
the ideas we give may turn out to be useful in investigating more traditional questions about QAOA as well. Because we have touched upon multiple approaches, 
we tried to be brief with each topic. For brevity we have also opted for leaving out some observations, 
for instance ones that concerned cases where the angles were randomly chosen from $[0,x]$ rather 
than from $[0,2\pi]$. Nevertheless, we hope no major information is missing from the article, and that 
some ideas within it will induce further wide-ranging QAOA studies. 

\section{Acknowledgements}
We are greatly indebted to Cupjin Huang, David Ding, David Gosset, Jianxin Chen, Yaoyun Shi and Ronald de Wolf
for their remarks, contributions and suggestions.

\bibliographystyle{alpha}
\bibliography{biblio}

\newpage
\part{Appendix}

\section{The Uncut-Polynomial}

Polynomials made from graphs, such as the Chromatic and Tutte polynomials, are frequent tools in graph theory.
We introduce a polynomial, related to these, which lets us conveniently think about QAOA-related tensor-networks.
Among the applications of our polynomial are the Triangle theorem and the derivation of formulas in Sections \ref{levelone} and \ref{leveltwo}.

\begin{definition}[$U$ polynomial]
Let $G$ be an undirected graph with edge set $\Ed$ and vertex set $\V$, where we allow loops and parallel edges. 
Let $\{ x_{e} | e \in \Ed\}$ be a variable set assigned to the edges of $G$.
If $c: \V\rightarrow \{0,1\}$ (such a function is called a {\em cut}) and $e\in \Ed$, we write $e\prec c$ to denote that $c$ gives the same value to both end points of $e$.
If $e$ is a loop then $e \prec c$ is automatic.
 For $c: \V\rightarrow \{0,1\}$ define
\[
X_{c} = \prod_{e: \, e\prec c} x_{e} 
\]
Let $v \in \V$ be an arbitrary vertex of $G$. The $U$ polynomial of $G$, {\em defined through $v$}, is 
\[
U(G) \;\; = \;\; \sum_{\begin{array}{c} c: \V\rightarrow \{0,1\} \\ c(v) = 0 \end{array}} X_{c}  \;\;\;\; = \;\;\;\; {1\over 2}  \sum_{\begin{array}{c} c: \V\rightarrow \{0,1\}\end{array}} X_{c}
\]
\end{definition}

\noindent The definition is independent of $v$, and the following facts are not hard to show:
\begin{enumerate}
\item $T$ is a tree if and only if $U(T) = \prod_{e\in \Ed(T)} (x_{e} + 1)$.
\item For any connected graph $G$ on $n$ nodes $U(G)$ contains 
exactly $2^{n-1}$ terms (monomials), each with coefficient one.
\item Let $C_{n}$ be the cycle on $n$ node, with edges labeled with $\Ed = \{1,\ldots, n\}$.
Then 
\[
U(C_{n}) = \sum_{\begin{array}{c} S\subseteq n\\ |S| = n \mod 2 \end{array}} \prod_{i\in S} x_{i}
\]
\end{enumerate}

\section{Annulling Rules and identities}

We shall now discuss rules that can be expressed in such a fashion, that certain variable replacements
make $U(G)$ identically zero, and we name them
{\em Annulling Rules}. We also discuss other identities. First agree on a notation: When $(u,v)$ is an edge of a graph, and 
we make a replacement $x_{u,v} \leftarrow A$ in the $U$ polinomial of $G$, then we denote this with
\[
U(G,(u,v):A) 
\]
We cal also put multiple replacements into the argument, each of the form $e: A$, where $e$ is an edge and $A$ 
is a value. Our first lemma 
describes one of the simplest annulling rules:

\begin{lemma}\label{bridge}
Let $G$ be a graph, and $e$ a bridge in $G$. Then $U(G, e : -1) =0$.
\end{lemma}

This lemma immediately follows from two simple facts: Fact 1.
When $G$ consists of a single edge, $e$, then $U(G) = x_{e}+1$, so $U(G, e : -1) =0$. Fact 2.
The following lemma:

\begin{lemma}
Let graphs $G_{1}$ and $G_{2}$ share a single node $v$, but do not share loops on $v$. Then 
\[
U(G_{1} \cup G_{2}) = U(G_{1}) \,U(G_{2}) 
\]
\end{lemma}
\begin{proof}
Define $U(G)$, $U(G_{1})$ and $U(G_{2})$ through the vertex $v$. Comparing these expressions
there will be a one-one correspondence between the terms of $U(G)$ and pairs of terms with first and second components from
$U(G_{1})$ and $U(G_{2})$ respectively. \end{proof}

\noindent {\em Note:} For {\em vertex disjoint} graphs, $G_{1}$ and $G_{2}$, 
an extra factor of 2 comes in: Let $G_{1}$ and $G_{2}$ be vertex disjoint. Then
\[
 U(G_{1} \cup G_{2}) = 2 \, U(G_{1}) \, U(G_{2}).
 \]
  
If a graph $G$ is a union of two graphs that intersect in two nodes,
we still have an expression of the $U$-polynomial of $G$ in terms of its components. First a definition:

\begin{definition}
We say that $G'$ arises from $G$ by identifying two nodes, $v, w\in \V(G)$ 
(i.e. merging them into a single node), if all edges between $v$ and $w$ become loops,
and for any node $x$ all edges from $x$ to $v$ and from $x$ to $w$ become parallel edges between $x$ and the 
new node (i.e. we do not merge edges). 
\end{definition}

\begin{lemma}\label{twonodes}
Let graphs $G_{1}$ and $G_{2}$ share two nodes, $u$ and $v$, but not edges. Then 
\[
U(G_{1} \cup G_{2}) = U(G'_{1}) \, U(G'_{2}) \;+\;  U(G''_{1}, (u,v): 0) \; U(G''_{2}, (u,v): 0)
\]
where we get $G'_{1}$ from $G_{1}$ by
identifying nodes $u$ and $v$,
and $G''_{1}$ from $G_{1}$ by adding an extra edge, $(u,v)$ to $G_{1}$. 
We get $G'_{2}$ and $G''_{2}$ from $G_{2}$ similarly.
\end{lemma}
The lemma holds even with loops and single or parallel edges on $\{u,v\}$ in $G_{1}$, $G_{2}$ or both, via
\begin{equation}\label{decompose}
U(G_{1}\cup G_{2}) = \sum_{
\begin{array}{l}
c\in {\rm CUT}(G_{1}\cup G_{2}) \\
c(u) = c(v) = 0
\end{array}
} X_{c} \;\; + \;
\sum_{
\begin{array}{l}
c\in {\rm CUT}(G_{1}\cup G_{2}) \\
c(u) = 0, \; c(v) = 1
\end{array}
} X_{c} \;\;\; = S' + S''
\end{equation}

\noindent $\bullet\;$ It is clear that $U(G'_{1}) \, U(G'_{2}) = S'$. \\
$\bullet\;$ That $U(G''_{1}, (u,v) :  0) \; U(G''_{2}, (u,v) : 0) = S''$,
follows from that 
$x_{(u,v)}\leftarrow 0$ sets all those terms of $U(G''_{1})$ and $U(G''_{2})$ zero
that belong to cuts $c$ with $c(u) =  c(v) = 0$. More generally:

\begin{lemma}
Let $G$ be a graph with two distinct nodes $u$ and $v$. Then
\[
U(G) = U(G') + U(G'', (u,v) : 0)
\]
where we get $G'$ from $G$ by
identifying nodes $u$ and $v$ and $G''$ from $G$ by 
adding an extra edge, $(u,v)$, to $G$.
\end{lemma}

In the following lemma we use the notation 
$G - e$ for the graph $(\V(G), \Ed(G)\setminus \{e\})$.

\begin{lemma}\label{trivial}
\smallskip
\begin{enumerate}
\item Let $e$ be an edge of graph $G$. 
Then 
\[
U(G, e: 1) = U(G - e)
\]
\item Let $e$ and $f$ be parallel edges in a graph $G$. 
\[
U(G, e:  C, f : D) = U(G - f, e:  CD)
\]
\end{enumerate}
\end{lemma}

We are ready to prove an annulling rule, which implies our Triangle Theorem.
We only prove the annulling rule, without the Triangle Theorem implication.

\begin{figure}[H]
\centering
\includegraphics[width=0.5\textwidth]{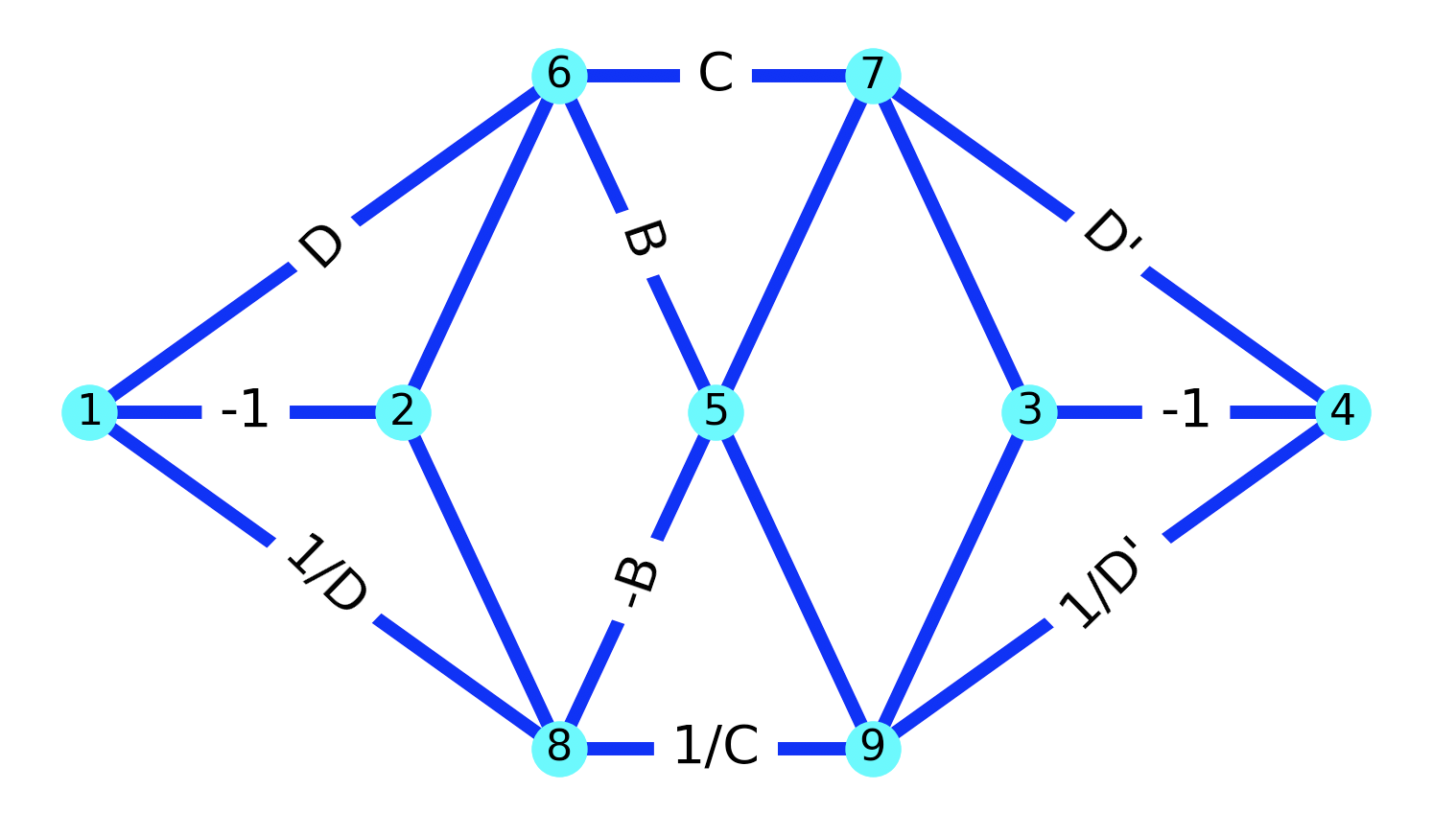}
\caption{$U(G)$ of the graph $G$ in the figure with replacements indicated on the edges
gives zero for every setting of the unlabeled edges. \label{zero0}}
\end{figure}

\begin{lemma}\label{anulemma}
Let $C$, $D$ and $D'$ be arbitrary non-zero constants and $B$ be an arbitrary constant. Then the $U$ polynomial of the graph $G$ in Figure \ref{zero0}
with replacements of the variables associated to its edges as drawn, gives a non-zero polynomial.
\end{lemma} 

\begin{proof}
We decompose $G$ into graphs $G_{L}$ and $G_{R}$, induced on node sets 
\[
\V_{L} = \{1,2,6,8\} \;\;\;\; {\rm and} \;\;\;\; \V_{R} = \{6,8,7,5,9,3,4\}
\]
Then we apply Lemma \ref{twonodes}, since $G_{L}$ and $G_{R}$ share exactly two nodes, 6 and 8.
First notice that the term $U(G'_{L})U(G'_{R})$ becomes zero after the replacements. For this we prove that  $U(G'_{L})$ becomes zero. 

\begin{figure}[H]
\centering
\begin{tabular}{ccc}
\includegraphics[width=0.2\textwidth]{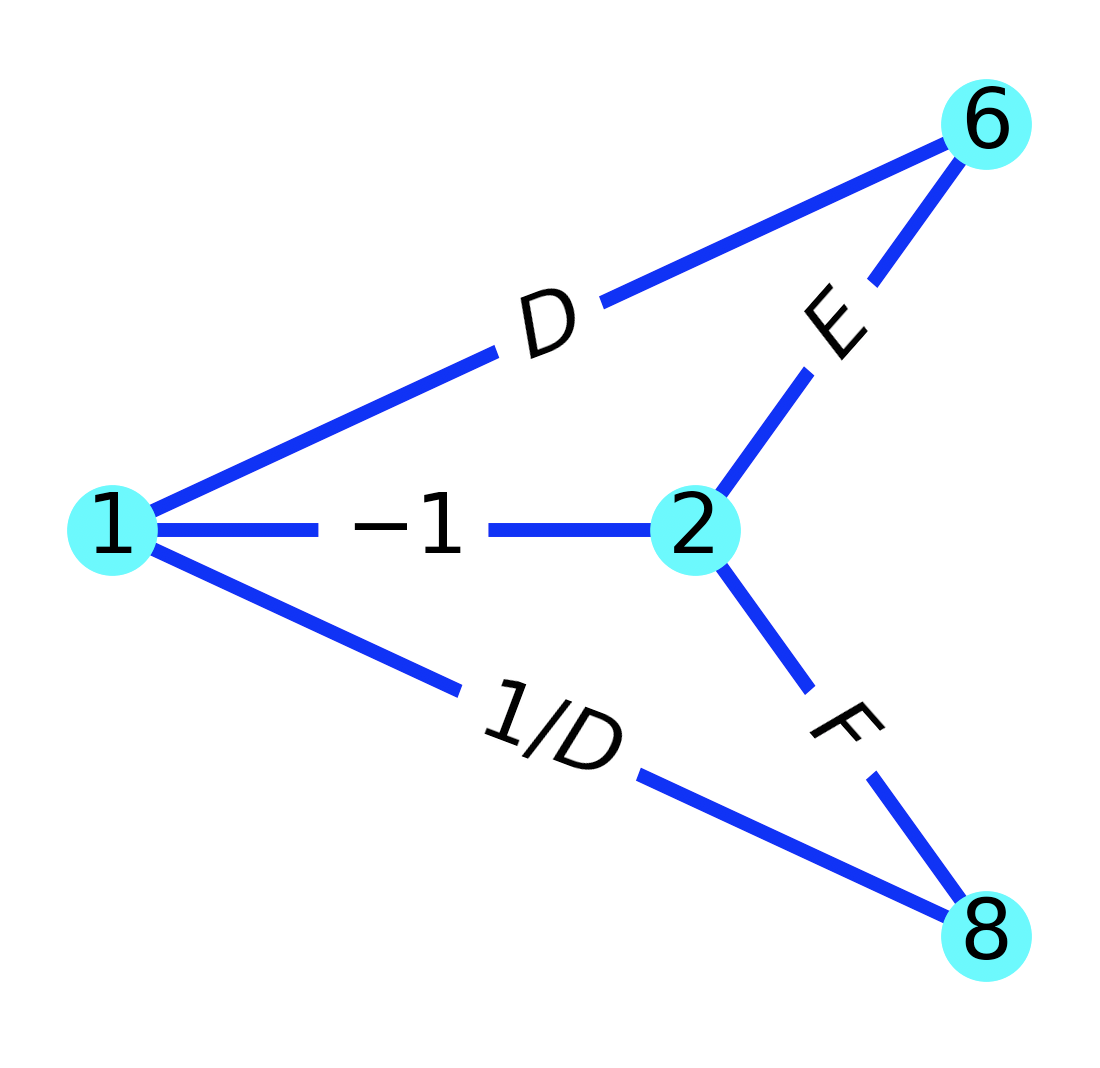}
\includegraphics[width=0.15\textwidth]{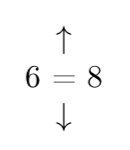}
 & \includegraphics[width=0.18\textwidth]{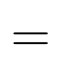} & \includegraphics[width=0.27\textwidth]{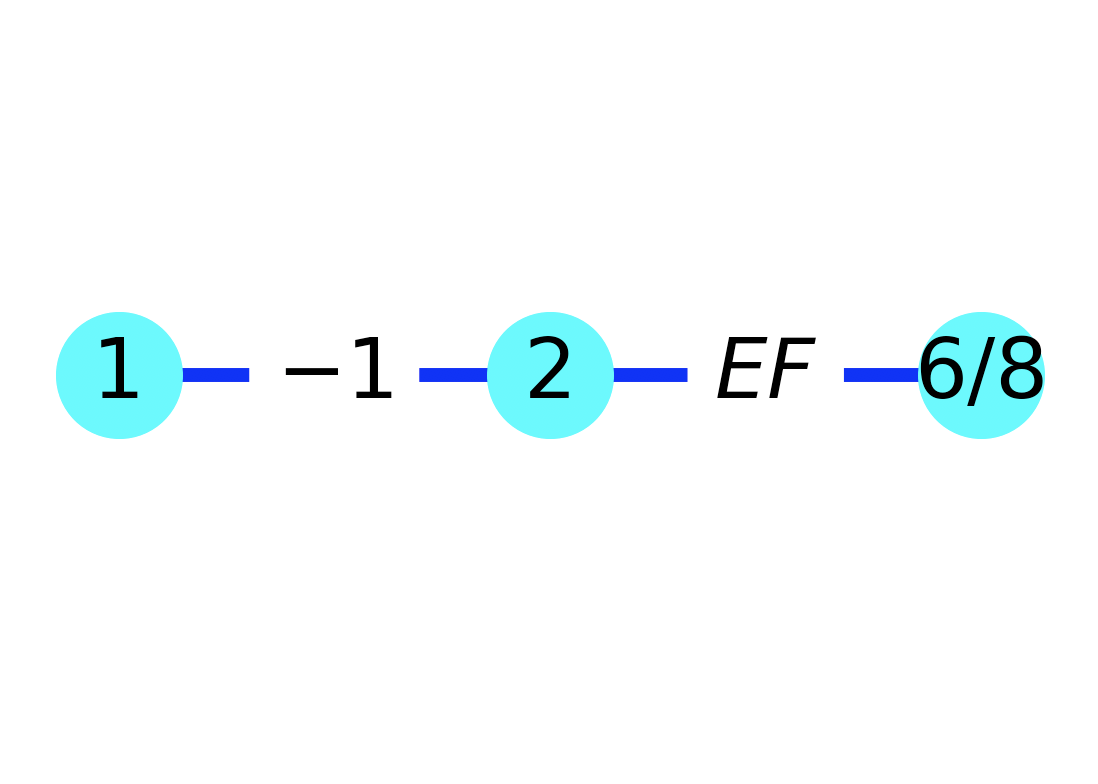} \\
\end{tabular}
\caption{Merging nodes 6 and 8  \label{zero55}}
\end{figure}

Graph $G'_{L}$ has node set  $\{1, 2, 6/8\}\}$, where node 6/8 arises from merging nodes 6 and 8 of $G_{L}$
The two edges, $(1,6)$ and $(1,8)$ become parallel.  Lemma \ref{trivial} allows to merge, then delete these edges,
and what we end up with is a path of length two as in the r.h.s. of Figure \ref{zero55}.
We finally apply Lemma \ref{bridge} to show that $U$ polynomial of this graph 
with the replacement as indicated, becomes zero, as $(1,2)$ becomes a bridge.

\begin{figure}[H]
\centering
\begin{tabular}{ccc}
\includegraphics[width=0.5\textwidth]{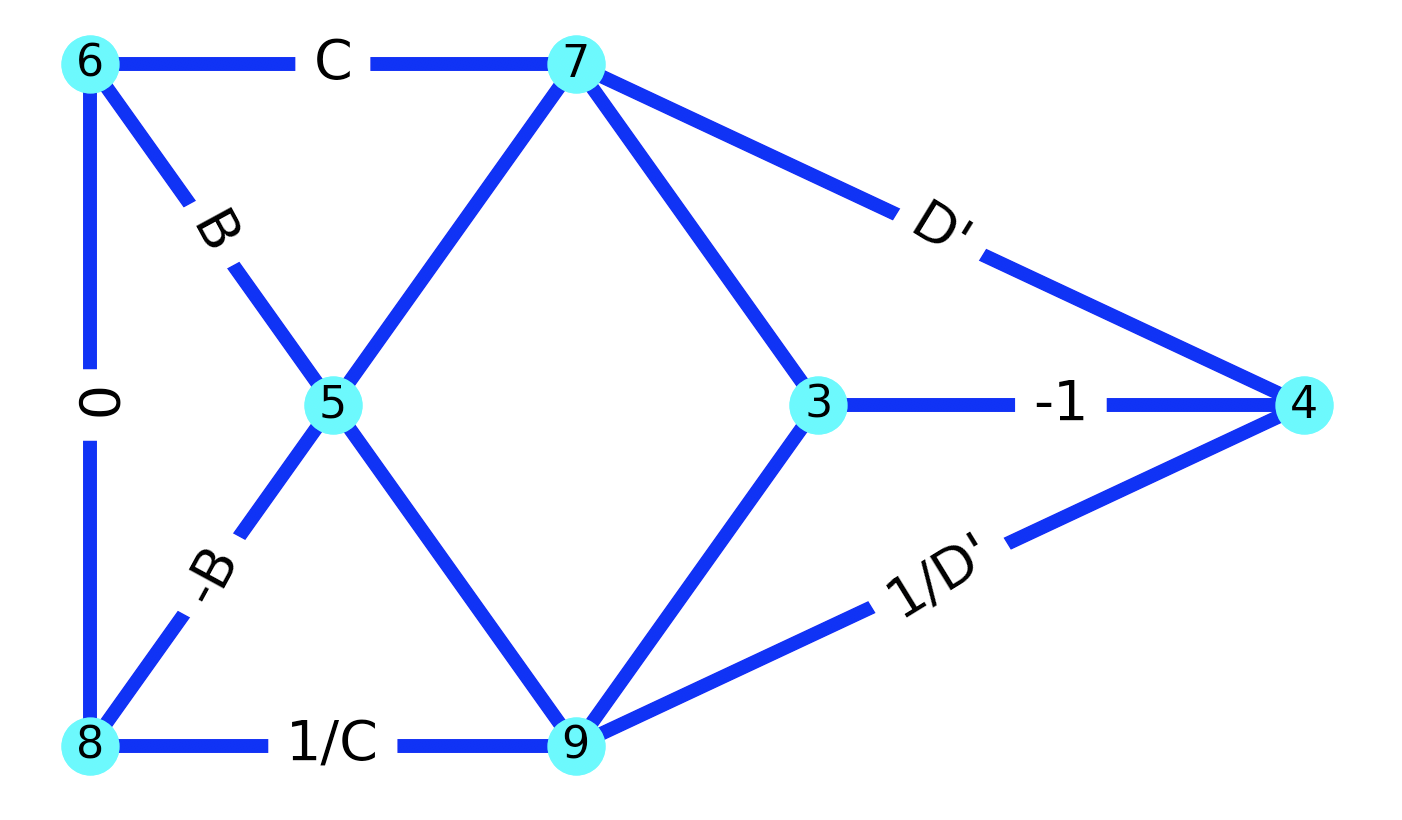} & \hspace{0.1in}  & \includegraphics[width=0.3\textwidth]{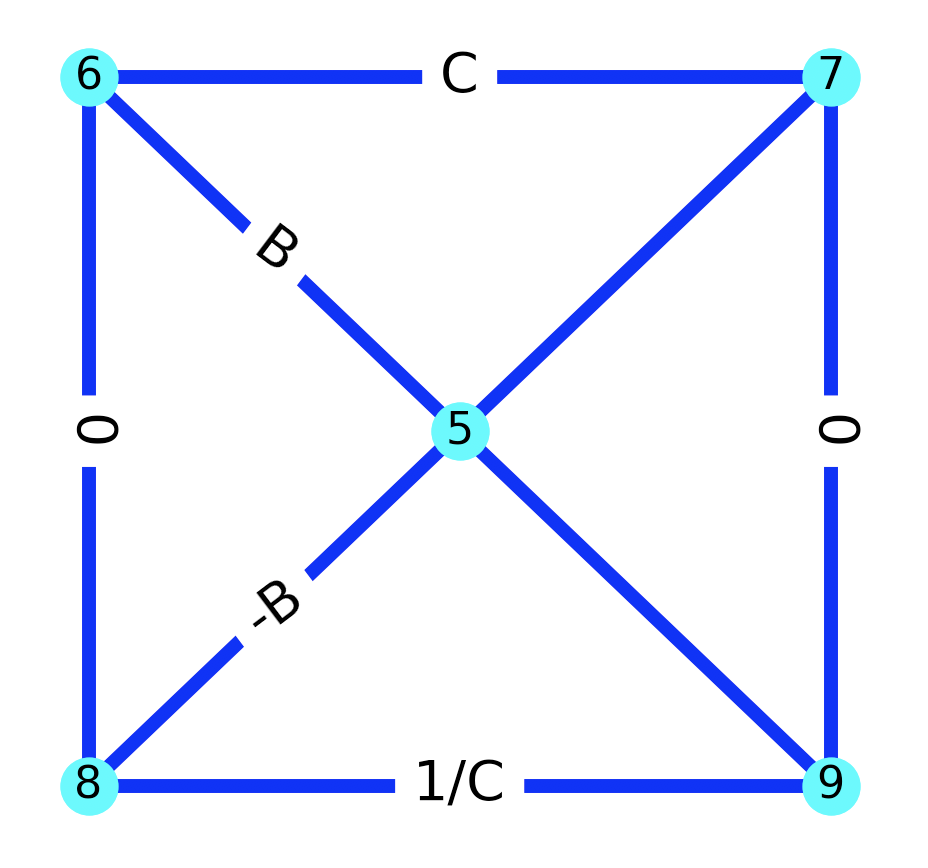}  \\
$G''_{R}$ & & $G_{\boxtimes}$
\end{tabular}
\caption{The graph $G''_{R}$ with all the replacements and $G_{\boxtimes}$ with all the replacements.\label{zero1}}
\end{figure}

Next we show that $U(G''_{L}, (6,8): 0) \; U(G''_{R}, (6,8) : 0) = 0$ by showing that
$U(G''_{R}, (6,8): 0) = 0$.  Figure \ref{zero1} shows $G''_{R}$ with all the replacements, including $x_{(6,8)}\leftarrow 0$. Just as we 
have decomposed $G$ through nodes 6 and 8, we now decompose $G''_{R}$ through nodes 7 and 9, getting two graphs induced on node sets
\[
\V_{\boxtimes} = \{6,7,8,9\} \;\;\;\; {\rm and} \;\;\;\; \V_{\rsub} = \{7,3,4,9\}
\]
Applying  Lemma \ref{twonodes} again for this decomposition (we do not write out the entire expression), we find
that the first summand of the r.h.s. is zero for the same reason as $U(G'_{L})$ was zero. Thus, by looking at the other summand, we notice that
it is sufficient to show that the $U$ polynomial of $G_{\boxtimes}$ on the right side of Figure \ref{zero1}, with the replacements as shown
on the edges, is zero. This is what we shall do below.

Recall (or realize) that if the value of an edge is zero, only those cuts create non-zero terms, where the two end-points 
of the edge are evaluated differently by the cut. We compute 
$U(G_{\boxtimes})$ with the replacements shown in Figure \ref{zero1}. 
We define the $U$ polynomial of $G_{\boxtimes}$ through node 6. 
This fixes $c(6) = 0$ for all cuts $c$ in te sum, and also $c(8)=1$. We also have $c(9) = 1-c(7)$. This leaves us with four non-zero terms:

\medskip
\begin{center}
\begin{tabular}{|cc|rrr|}\hline
$c(7)$ & $c(5)$ & \multicolumn{3}{c}{associated term} |\\[6pt]\hline\hline\\[-1.5ex]
0 & 0 & $B\cdot x_{(5,7)}\cdot C / C$ & = & $B\cdot x_{(5,7)}$ \\
0 & 1 & $-B\cdot x_{(5,9)} \cdot C / C$ & = & $-B\cdot x_{(5,9)}$ \\
1 & 0 & & & $B\cdot x_{(5,9)}$ \\
1 & 1 & & & $-B\cdot x_{(5,7)}$ \\\hline
\end{tabular}
\end{center}
\medskip
Thus with ${\rm LIST} \;= \; (6,7):C, \; (6,8):0, \; (6,5):B, \; (8,5):-B, \; (8,9): 1/C,\; (7,9):0$ we have:
\[
U(G_{\boxtimes}, {\rm LIST}) 
= B x_{(5,7)}-B x_{(5,9)}-B x_{(5,7)}+B x_{(5,9)} = 0
\]
\end{proof}

\section{QAOA energy, computed with the Uncut polynomial}

\begin{figure}[H]
\centering
\includegraphics[width=0.6\textwidth]{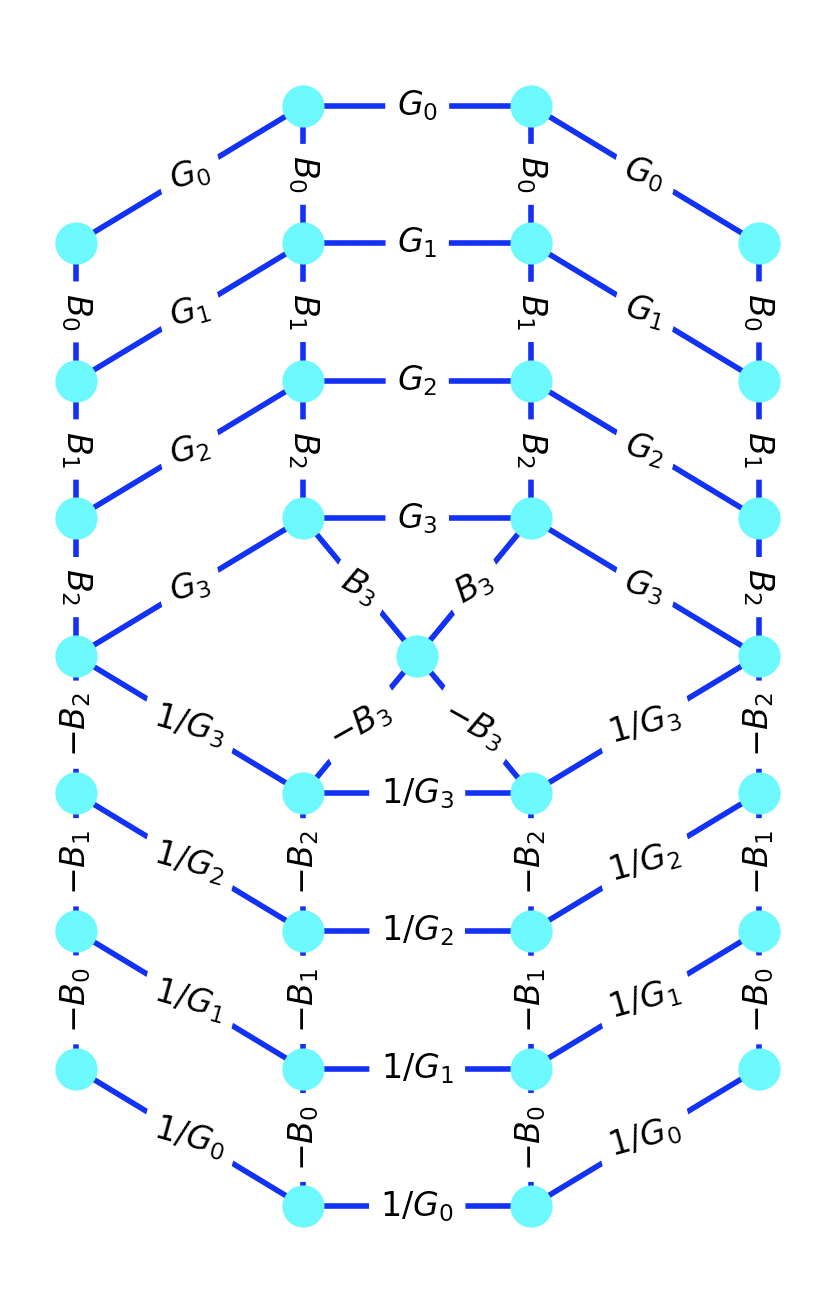} 
\caption{Computing the level 4 QAOA energy of the middle edge $e$ of the path of length 3: $\bullet-\bullet\stackrel{e}{-}\bullet-\bullet$ \\
In the $U$ polynomial of the above graph we need to replace variables as shown on the edges, where $G_{j} = e^{-i\gamma_{j}}$ and $B_{j} = -i \tan \beta_{j}$. 
The obtained value must be further multiplied with ${1\over 2^{3}}\prod_{q=0}^{3} \cos^{4}\beta_{q}$. \label{qaoa1}}
\end{figure}

\medskip

The QAQA energy of a graph $G$ is the sum of the QAOA energies of its edges.
We can compute the QAOA energy of an edge $e$ of $G$ from the uncut polynomial of the graph
${\cal G}(G,e,p)$, which is constructed from $G$, $e$ and the number $p$ of levels. In the definitions below we fix
$G$, and omit it from most notations.

The zeroth edge-neighborhood of $e$ in $G$, denoted by $N_{0}(e)$, consists only of $e$. For $i > 0$ the $i^{\rm th}$ edge neighborhood, $N_{i}(e)$ of $e$, consists of 
all elements of $N_{i}(e)$ and of all edges that are incident to any edge in $N_{i}(e)$. We also define $V_{i}(e)$ as the set of nodes that are incident
to any of the edges in $N_{i}(e)$. For the number of levels, $p$, for graph $G$ and for edge $e$ we define a graph ${\cal G}(G,e,p)$ with vertex set:
\[
\V({\cal G}(G,e,p)) = \{0\} \cup \bigcup_{q=0}^{p-1} {\cal N}_{q}  \cup {\cal N}'_{q}  \cup {\cal M}_{q}
\]

Both ${\cal N}_{q}$ and ${\cal N}'_{q}$ are copies of the set $V_{q}(e)$. The set ${\cal M}_{q}$
is the copy of the set $V_{q+1}(e) \setminus V_{q}(e)$.

\medskip

\noindent The graph ${\cal G}(G,e,p)$ has several classes of labeled edges. The labels correspond to future replacements.

\begin{enumerate}
\item The node 0 of ${\cal G}(G,e,p)$ has four incident edges: two towards ${\cal N}_{0}$, labeled with $B_{p-1}$, and two towards ${\cal N}'_{0}$, labeled with $-B_{p-1}$.
\item  For $0\le q \le p-1$ the set ${\cal N}_{q}  \cup {\cal M}_{q}$ is a copy of $V_{q+1}$. Copy all edges 
of $N_{q+1}(e)$ onto  ${\cal N}_{q}  \cup {\cal M}_{q}$ and label each with $G_{p-1-q}$.
\item  For $0\le q \le p-1$ the set ${\cal N}'_{q}  \cup {\cal M}_{q}$ is a copy of $V_{q+1}$. Copy all edges 
of $N_{q+1}(e)$ onto  ${\cal N}'_{q}  \cup {\cal M}_{q}$ and label each with $1/G_{p-1-q}$.
\item For $0\le q \le p-2$ there is a natural matching between ${\cal N}_{q+1}$ and ${\cal N}_{q}  \cup {\cal M}_{q}$. Label all edges of this matching with $B_{p-2-q}$.
\item For $0\le q \le p-2$ there is a natural matching between ${\cal N}'_{q+1}$ and ${\cal N}'_{q}  \cup {\cal M}_{q}$. Label all edges of this matching with $-B_{p-2-q}$.
\end{enumerate}

\medskip

\noindent  Turn the labels into formulas by the replacements
\begin{eqnarray*}
B_{q} \; & \rightarrow &  \; -i \tan \beta_{q} \\
G_{q} \; & \rightarrow &  \; e^{-i\gamma_{q}} \\
\end{eqnarray*}

\noindent  After we compute the $U$ polynomial of ${\cal G}(G,e,p)$, where each variable is replaced with the expression  
labeling the associated edge, we multiply the result with 
\[
{1\over 2^{|V_{p}(e)|-1}}\prod_{q=0}^{p-1} \cos^{2 \, |V_{q}(e)|}\beta_{p-1-q}
\]

\begin{figure}[H]
\centering
\includegraphics[width=0.9\textwidth]{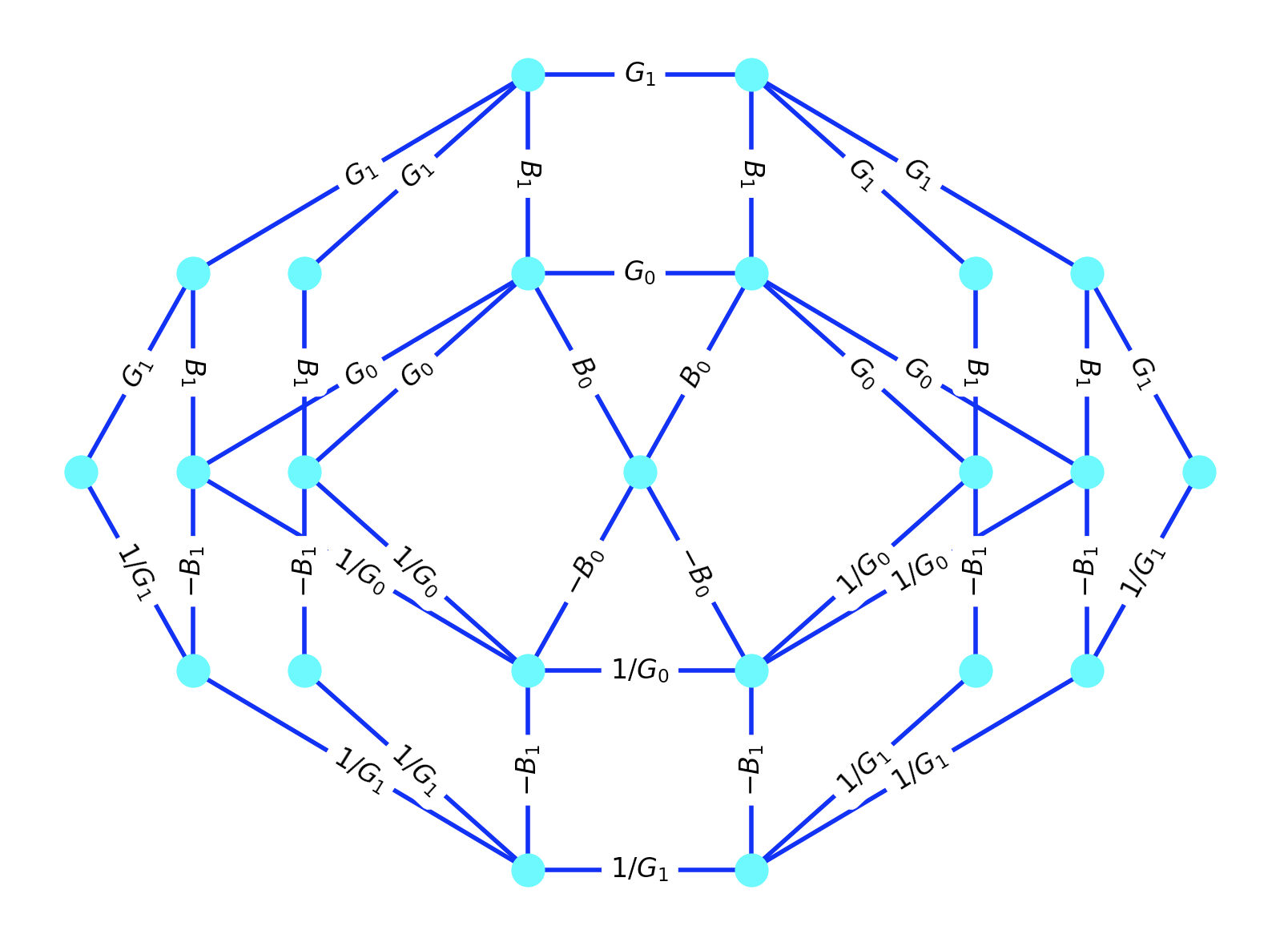}
\caption{Level 2 QAOA of $\; G \; = \; \bullet-\bullet-\stackrel{\stackrel{\bullet}{|}}{\bullet}\stackrel{e}{-}\stackrel{\stackrel{\bullet}{|}}{\bullet}-\bullet-\bullet\;\;\;$ 
as a $U$ polynomial. \label{bee}}
\end{figure}

\section{A density matrix example}

%\begin{figure}[H]
%\centering
%\includegraphics[width=1\textwidth]{sigma2edge} 
%\caption{$\sigma^{(2)}_{1}$ of a single edge. Zeros are not printed. \label{z21}}
%\end{figure}

%\medskip

\begin{figure}[H]
\centering
\includegraphics[width=0.6\textwidth]{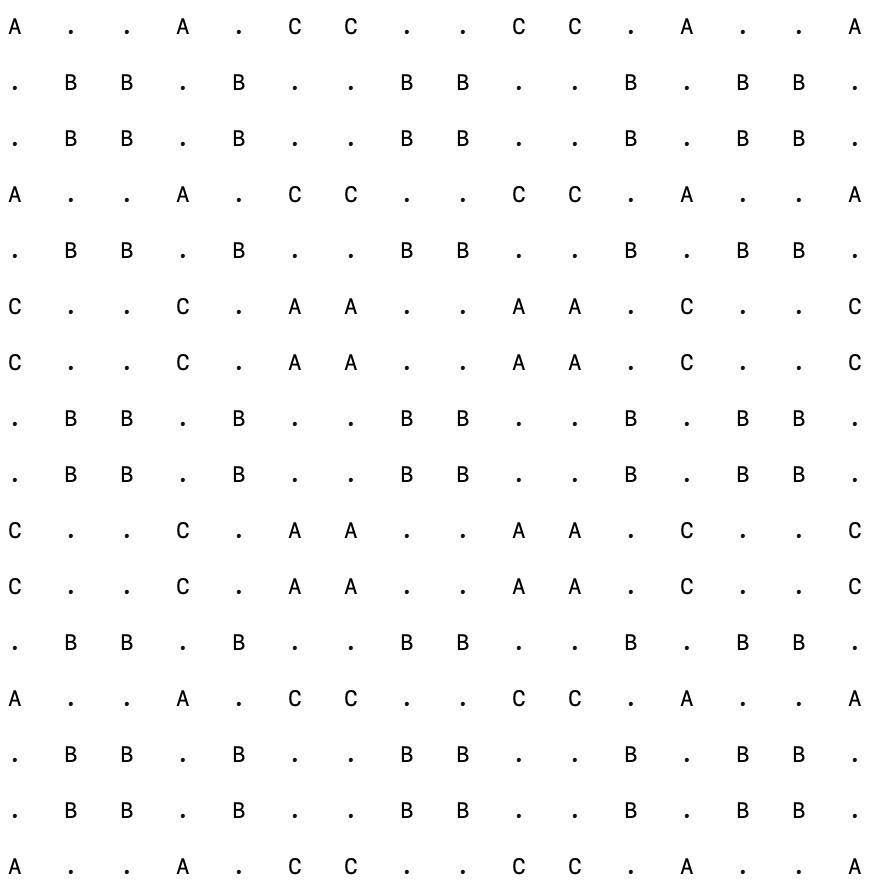} 
\caption{The structure of $\sigma^{(2)}_{1}$ for $G=(\{a,b\},\{e\})$ ; $\;\;\;\;\;A = 0.078125$, $B=0.046875$, $C = 0.015625$. \label{pattern}}
\end{figure}

\medskip

\noindent Let $C$ for $e = \langle ab\rangle$ be the Hamiltonian of $G=(\{a,b\},\{e\})$ (Single Edge Graph):
\[
C = {1\over 2}(1 + \sigma_{a}^{z}\sigma_{b}^{z})
\]

\noindent We calculate the evolution of the second order density matrices and associated variances. Matrix $\sigma^{(2)}_{1}$ has the pattern shown in 
Figure \ref{pattern}, and one can observe the same pattern for $\sigma^{(2)}_{2}$, $\sigma^{(2)}_{3}$, etc.
Let $A_{p}, B_{p}$ and $C_{p}$ be the three parameters defining $\sigma^{(2)}_{p}$. 
From the explicit expression in Section \ref{ssquare} one can show the recurrence

\begin{eqnarray*}
A_{p} & = & 0.75 A_{p-1} + 0.5 B_{p-1}  \\
B_{p} & = & 0.25 A_{p-1} + 0.5 B_{p-1}  \\
C_{p} & = & -0.25 A_{p-1} + 0.5 B_{p-1} 
\end{eqnarray*}

Writing the recurrence differently,

\[
\left(
\begin{array}{c}
A_{p} \\
B_{b}
\end{array}
\right) = 
\left(
\begin{array}{cc}
0.75 & 0.5 \\
0.25 & 0.5
\end{array}
\right)
\left(
\begin{array}{c}
A_{p-1} \\
B_{p-1}
\end{array}
\right)
\]

We immediately notice that $A_{p} +B_{p} = A_{p-1} +B_{p-1}$, so the sum always remains $1/8 = A_{0} + B_{0}$.
From $B_{p} = 1/8 - A_{p}$ we get the recurrence $A_{p} = A_{p-1}/4 + 1/16$, leading to 
\[
A_{p} =  {1\over 16} \left( 1 + {1\over 4 }  \cdots + {1\over 4^{p} }\right) = {1- 0.25^{p+1} \over 12} 
\]

Since for any $z\in \{0,1\}^{2}$ we have
\[
C(z) =
\left\{
\begin{array}{lll}
1 & {\rm if} & z \in \{00,11\} \\
0 & {\rm if} & {\rm otherwise}
\end{array}
\right.
\]
and 
\[
\sigma^{(2)}_{p}[(00, 00),(00, 00)] = \sigma^{(2)}_{p}[(00, 11),(00, 11)] = \sigma^{(2)}_{p}[(11, 00),(11, 00)] = \sigma^{(2)}_{p}[(11, 11),(11, 11)] = A_{p}
\]

we have that 

\medskip

\mybox{lightgray}{
\[
{\bf Tr} ({C}^{\otimes 2}\sigma^{(2)}_{p}) = \sum_{z_{1},z_{2} \in\{0,1\}^{2}} C(z_{1})C(z_{2} ) \, \sigma^{(2)}_{p}[(z_{1}, z_{2})(z_{1}, z_{2})] \; = \; 4 A_{p} \; =  \; (1- 1/4^{p+1} )/ 3
\]
}

\bigskip

By Lemma \ref{edgesigma} we have that ${\bf Tr} ({C}\sigma_{p})$ is 0.5 for all $p$. Therefore the variance of the QAOA energy 
for the Single Edge Graph (over random angles) for level $p$ is 

\[
 {\bf Tr} ({C}^{\otimes 2}\sigma^{(2)}_{p}) - {\bf Tr} ({C}\sigma_{p})^{2} = {1 \over 12} - {1\over 3\times 4^{p+1} }
\]

\bigskip

{\small
\begin{center}
\begin{tabular}{|ccccc|}\hline\hline
level ($p$) & 1 & 2 & 3 & 4    \\[3pt]
variance & 0.0625 & 0.078125 & 0.08203125 & 0.0830078125 \\\hline
\end{tabular}
\end{center}}

\end{document}